\documentclass[final,5p,times,twocolumn]{elsarticle}

\usepackage[english]{babel}

\usepackage{ifpdf}

\usepackage{lineno,hyperref} 
    \modulolinenumbers[1]

\usepackage{graphicx}
\usepackage{color}
\usepackage{placeins}

\usepackage[cmex10]{amsmath}
\usepackage{amsfonts, amssymb, amsthm}
\usepackage{mathrsfs}

\usepackage{subfigure}
\usepackage{algorithm,algorithmic}

\usepackage{hyphenat}
\usepackage{tikz,chemarrow}
\usepackage[framemethod=tikz]{mdframed}
\usepackage[normalem]{ulem}

\usepackage{dsfont}

\input{mysymbol.sty}

\newtheorem{mytheorem}{Theorem}

\newtheorem{mydefinition}{Definition}
\newtheorem{myproposition}{Proposition}

\newtheorem{mylemma}{Lemma}
\newtheorem{myassumption}{Assumption}

\journal{Signal Processing}

\begin{document}

\begin{frontmatter}

\title{Blind Deconvolution on Graphs: Exact and Stable Recovery\tnoteref{funding,conferences}}
\tnotetext[funding]{This work was supported in part by the National Science Foundation under award ECCS-1809356 and by the Center of Excellence in Data Science, an Empire State Development-designated Center of Excellence.}
\tnotetext[conferences]{Part of the results in this paper were presented at the \textit{2018 EUSIPCO} conference~\cite{chang2018eusipco}.}

\author[uofr]{Chang Ye}
\ead{cye7@ur.rochester.edu}
\cortext[cor1]{Corresponding author}

\author[uofr,goergen]{Gonzalo Mateos\corref{cor1}}
\ead{gmateosb@ece.rochester.edu}

\address[uofr]{Dept. of Electrical and Computer Engineering, University of Rochester, Rochester, USA}
\address[goergen]{Goergen Institute for Data Science, University of Rochester, Rochester, USA}


\begin{abstract}
We study a blind deconvolution problem on graphs, which arises in the context of localizing a few sources that diffuse over networks.  While the observations are bilinear functions of the unknown graph filter coefficients and sparse input signals, a mild requirement on invertibility of the diffusion filter enables an efficient convex relaxation leading to a linear programming formulation that can be tackled with off-the-shelf solvers. Under the Bernoulli-Gaussian model for the inputs, we derive sufficient exact recovery conditions in the noise-free setting. A stable recovery result is then established, ensuring the estimation error remains manageable even when the observations are corrupted by a small amount of noise. Numerical tests with synthetic and real-world network data illustrate the merits
of the proposed algorithm, its robustness to noise as well as the benefits of leveraging multiple signals to aid the (blind) localization of sources of diffusion. At a fundamental level, the results presented here broaden the scope of classical blind deconvolution of (spatio-)temporal signals to irregular graph domains.
\end{abstract}



\end{frontmatter}



\section{Introduction}\label{S:Introduction}


Network processes such as neural activities at different cortical brain regions~\cite{weiyu_brain_signals,hu2016localizing,yang2022tsipn}, vehicle flows over transportation networks~\cite{deri2016new,hasanzadeh2019traffic}, COVID-19 infections across demographic areas connected via a commute flow mobility graph~\cite{yang2021icassp}, or spatial temperature
profiles monitored by distributed sensors~\cite{D_LMS_TSP,stankovic2019spmag}, can be represented as signals supported on the nodes of a graph. In this context, the graph signal processing (GSP) paradigm hinges on recognizing that signal properties are shaped by the underlying graph topology (e.g., in a network diffusion or percolation process), to develop models, signal representations, and information processing algorithms that exploit this relational structure. Accordingly, generalizations of key signal processing tasks have been widely explored in recent work; see~\cite{gsp2018tutorial,geert2023spmag} for recent tutorial accounts. Notably graph filters were conceived as information-processing operators acting on graph-valued signals~\cite{sandryhaila2013discrete,isufi2024gf}, and they are central to graph convolutional neural network models; see e.g.~\cite{gama2020spmag}. Mathematically, graph filters are linear transformations that can be expressed as polynomials of the so-termed graph-shift operator (GSO; see Section \ref{S:prelim}).
The GSO offers an alegbraic representation of network structure and can be viewed as a local diffusion operator. Its spectral decomposition can be used to represent signals and filters in the graph frequency domain~\cite{sandryhaila2014discrete}. For the cycle digraph representing e.g., periodic temporal signals, the GSO boils down to the time-shift operator~\cite{sandryhaila2013discrete,gsp2018tutorial,mateos2019connecting}.
Given a GSO, the polynomial coefficients fully determine the graph filter and are referred to as filter coefficients. 

\noindent \textbf{Problem description.} In this paper, we revisit the blind deconvolution task for graph signals introduced in~\cite{segarra2015camsap}, with an emphasis on modeling diffusion processes and localization network diffusion sources. Specifically, given $P$ observations of graph signals $\{\bby_i \}_{i=1}^{P}$ that we model as outputs of a diffusion filter (i.e., a polynomial in a known GSO), we seek to jointly identify the filter coefficients $\bbh$ and the input signals $\{\bbx_i \}_{i=1}^P$ that generated the network observations.  Since the resulting bilinear inverse problem is ill-posed, we assume that the inputs are sparse -- a well-motivated setting when few seeding nodes (the sources) inject a signal that is diffused throughout a network~\cite{segarra2017blind}. Localizing sources of network diffusion is a challenging problem with applications in several fields, including sensor-based environmental monitoring, social networks, neural signal processing, or, epidemiology, too name a few. This inverse problem broadens the scope of classical blind deconvolution of temporal or spatial signals to graphs \cite{ahmed2014blind,levin2011understanding,wang2016blind}. 

\noindent \textbf{Related work and contributions.} A noteworthy approach was put forth in~\cite{segarra2017blind}, which casts the (bilinear) blind graph-filter identification task as a linear inverse problem in the ``lifted'' rank-one, row-sparse matrix $\bbx \bbh^\top$. While the rank and sparsity minimization algorithms in \cite{segarra2017blind,david2021sp} can successfully recover sparse inputs along with low-order graph filters, reliance on matrix lifting can hinder applicability to large graphs. Beyond this computational consideration, the overarching assumption of~\cite{segarra2017blind} is that the inputs $\{\bbx_i \}_{i=1}^P$ share a common support. Here instead we show how a mild requirement on invertibility of the graph filter facilitates an efficient convex relaxation for the multi-signal case with arbitrary supports (Section \ref{S:blind_ID}); see also \cite{wang2016blind} for a time-domain precursor as well as generalizations of~\cite{chang2018eusipco} in a supervised learning setting~\cite{chang2022eusipco} and a study of robustness to graph perturbations~\cite{victor2024icassp}. In Section~\ref{S:amb_uni}, we establish sufficient conditions under which the proposed convex estimator can exactly recover sparse input signals, assumed to adhere to a Bernoulli-Gaussian model. We also derive a stability result for the pragmatic setting where the graph signal observations are corrupted by additive noise. The analysis is challenging, since the
favorable (circulant) structure of time-domain filters in~\cite{wang2016blind} is no longer present in the network-centric setting dealt with here. Numerical tests with synthetic and real data corroborate the effectiveness of the proposed approach in recovering the sparse input signals (Section~\ref{S:simulation}). Concluding remarks are given in Section~\ref{S:conclusion}. We defer proofs and mathematical details as well as non-essential experimental details to the appendices.

Different from most existing model-based works dealing with source localization on graphs, e.g.,  \cite{zhang2016towards,pinto2012locating,sefer2016diffusion}, like~\cite{pena2016source} we advocate a GSP approach based on an admittedly simple forward (graph filtering) model. In our case, this simplification facilitates a thorough theoretical analysis of recovery performance, offering valuable insights. Often the models of diffusion are probabilistic in nature, and resulting maximum-likelihood source estimators can only be optimal for particular (e.g., tree) graphs~\cite{pinto2012locating}, or rendered
scalable under restrictive dependency assumptions~\cite{feizi2016network}. Relative to~\cite{pena2016source,hu2016localizing}, the proposed framework can accommodate signals defined on general undirected graphs and relies on a \emph{convex} estimator of the sparse sources of diffusion. Furthermore, the setup where multiple output signals are observed (each one corresponding to a
different sparse input), has not been thoroughly explored in convex-relaxation approaches to blind deconvolution
of (non-graph) signals, e.g.,~\cite{ahmed2014blind,LingBiConvexCS}; see \cite{wang2016blind} for a recent and inspiring alternative that we leverage here. 

Relative to the conference precursor~\cite{chang2018eusipco} that discussed inherent scaling and (non-cyclic) permutation ambiguities arising with some particular graphs as well as identifiability issues~\cite{li2015unified}, the exact and stable recovery guarantees in Section \ref{S:amb_uni} are significant contributions of this journal paper. Here we offer a comprehensive presentation along with full-blown technical details, we study robustness to noise, and provide a markedly expanded experimental evaluation anchored in our theoretical analyses.

\medskip\noindent\emph{Notation:} {Entries of a matrix $\bbX$ and a (column) vector $\bbx$ are denoted as $X_{ij}$ and $x_i$.} 
Operators $(\cdot)^\top$, $\E{\cdot}$, $\textrm{vec}[\cdot]$, $\sigma_{\max}[\cdot]$, $\circ$, $\otimes$ and $\odot$ stand for matrix transpose, expectation, matrix vectorization, largest singular value, Hadamard (entry-wise), Kronecker, and Khatri-Rao (column-wise Kronecker) products, respectively. 
{The diagonal matrix $\diag(\mathbf{x})$ has $(i,i)$th entry $x_i$}. 
The $n\times n$ identity matrix is represented by $\bbI_n$, while $\mathbf{0}_n$ stands for the $n\times 1$ vector of all zeros, and $\mathbf{0}_{n\times p}=\mathbf{0}_n\mathbf{0}_p^\top$. A similar convention is adopted for vectors and matrices of all ones. For matrix $\bbM\in\reals^{N\times k}$, we use $\text{span}\{\bbM\}:= \{\bbz\in\reals^N|\bbz = \bbM\bbx,\forall \bbx\in\reals^k\}$ to represent the linear subspace spanned by its columns. The notation $\| \bbX \|_{p,q}=\left(\sum_{j}(\sum_i |X_{ij}|^p)^{q/p}\right)^{1/q}$ stands for the elementwise matrix norm, $\|\bbX\|_F=\| \bbX \|_{2,2}$ denotes Frobenius norm, and $\|\bbX\|_{l\rightarrow 2},l=1,2$ are operator norms, i.e., the maximum $\ell_2$ norm of a column of $\bbX$ for $l=1$ and $\sigma_{\max}[\bbX]$ for $l=2$, respectively.

\section{Preliminaries and Problem Statement} \label{S:prelim}


We briefly review the necessary GSP background to introduce the observation model and formally state the problem.


\subsection{Graph signal processing background}\label{ssec:gsp_basics}


Consider a weighted and undirected graph denoted as $\ccalG=(\ccalV,\bbA)$, where $\ccalV:=\{1,\ldots, N\}$ comprises the vertex set. 
The symmetric graph adjacency matrix $\bbA \in \reals_+^{N \times N}$ has entries $A_{ij}=A_{ji}\geq 0$, that represent the weight of the edge $(i,j)$ between nodes $i$ and $j$. Naturally, if such edge $(i,j)$ is not present in $\ccalG$ then $A_{ij}=0$. We do not allow for self-loops, hence $A_{ii}=0$ for all $i\in \ccalV$. Directed graphs are important~\cite{marques2020digraphs}, but beyond the scope of this work.

As a general algebraic descriptor of graph connectivity, we henceforth rely on a symmetric GSO $\bbS \in \reals^{N \times N}$, inheriting and encoding the sparsity pattern of $\ccalG$ \cite{puschel2008asp,sandryhaila2013discrete}. That is, the only requirement for an admissible GSO $\bbS$ is that $S_{ij}=0$ when there is no edge connecting $i$ and $j$. Typical choices are $\bbA$ itself, the combinatorial graph Laplacian $\bbL=\bbA-\textrm{diag}(\bbA \cdot \mathbf{1}_N)$, or, their various degree-normalized forms~\cite{gavili2017shift,gsp2018tutorial} that we will adopt for our experiments in Section \ref{S:simulation}. Given that $\bbS$ is real and symmetric, it is always diagonalizable and we can decompose it as $\bbS=\bbV\bbLambda\bbV^\top$, where $\bbLambda=\textrm{diag}(\lambda_1,\ldots,\lambda_N)$ collects the eigenvalues and $\bbV$'s columns comprise an orthonormal basis of GSO eigenvectors. 

Lastly, we define a graph signal $\bbx: \ccalV \mapsto \reals^N$ as an $N$-dimensional vector, where $x_i$ denotes the signal value at node $i \in \ccalV$. This value could for instance represent the rating user $i$ assigns to a particular item (e.g., an article as in the experiments in Section \ref{ssec:real_data}), the measurement collected by sensor $i$, or, the neural activity recorded in the $i$-th cortical region-of-interest as defined by some brain parcellation. In GSP, the prevalent view is that graph signal $\bbx$ should not be viewed in isolation, but rather as a tuple along with the GSO $\bbS$ -- the later provides useful contextual (or prior) information about pairwise relationships among individual signal elements in $\bbx$. Next, we elaborate on the operator viewpoint of $\bbS$ as a map between graph signals.


\subsection{Graph-filter models of network diffusion}\label{ssec:model_diffusion}


Let $\bby$ represent a graph signal supported on $\ccalG$, which is assumed to be generated from an initial state $\bbx$ through linear network dynamics of the form
\begin{align}\label{eq:network_diffusion}
\textstyle \bby\ =\ \alpha_0 \prod_{l=1}^{\infty} (\bbI_N-\alpha_l \bbS) \bbx
\ =\ \sum _{l=0}^{\infty} \beta_l \bbS^l \bbx .
\end{align}
The linear transformation $\bbS$ represents one-hop network aggregation (or averaging), so repeated ($l=1,2,\ldots$) applications of the GSO in \eqref{eq:network_diffusion} diffuse $\bbx$ across $\ccalG$. Accordingly, $\{\alpha_l\}_{l=0}^\infty$ and $\{\beta_l\}_{l=0}^\infty$ can be viewed as diffusion coefficients for the multiplicative and additive signal model parametrizations in \eqref{eq:network_diffusion}, respectively. This generative model for $\bby$ is admittedly simple, but it nonetheless encompasses a broad class of linear network processes such as heat diffusion, average consensus, PageRank, and the DeGroot model of opinion dynamics~\cite{DeGrootConsensus,to2020blindcd}.

The Cayley-Hamilton theorem~\cite[p.109]{hornjohnson2012matrixanalysis} ensures that the infinite series in the right-hand-side of \eqref{eq:network_diffusion} can always be equivalently reparametrized using \emph{bounded-degree} polynomialss of $\bbS$. Introducing the coefficient vector $\bbh:=[h_0,\ldots,h_{L-1}]^\top$ and the shift-invariant graph filter~\cite{gsp2018tutorial}
\begin{equation}\label{eq:graph_filter_def}
\mathbf{H}:=h_0\bbI_N+h_1 \mathbf{S}+h_2 \mathbf{S}^2+\ldots+h_{L-1} \mathbf{S}^{L-1}=\sum_{l=0}^{L-1}h_l \mathbf{S}^l,
\end{equation}
the signal model \eqref{eq:network_diffusion} becomes
\begin{equation}\label{eq:graph_filter_model}
\bby = \left(\sum_{l=0}^{L-1}h_l \bbS^l\right)\bbx = \bbH \bbx,
\end{equation}
for some $\bbh$ and $L\leq N$. While graph filters are leveraged here as simple generative mechanisms to describe diffusion processes on networks, these convolutional operators play a central role in GSP and machine learning on graphs; see e.g.,~\cite{isufi2024gf} for a recent tutorial treatment. 

\vspace{2pt}
\noindent \textbf{Frequency representation.} Graph filters and signals admit representations in the frequency (i.e., graph spectral) domain~\cite{sandryhaila2014discrete,gsp2018tutorial}. To this end, recall the GSO eigenvalues $\lambda_1,\ldots,\lambda_N$ and introduce the $N\times L$ Vandermonde matrix $\bbPsi_L$, where $\Psi_{ij}:=\lambda_i^{j-1}$ and $L$ is the order of the graph filter \eqref{eq:graph_filter_def}. The graph Fourier transform (GFT) of a signal $\bbx$ and the frequency response of filter $\bbh$ are $\tbx:=\bbV^\top\bbx$ and $\tbh:=\bbPsi_L\bbh$, respectively. This follows by evaluating the GFT of the filter's output $\bby=\bbH\bbx$ and using the spectral decomposition $\bbS=\bbV\bbLambda\bbV^\top$, to yield
\begin{equation} \label{eq:freq_response}
\tby=\diag\big(\bbPsi_L\bbh\big)\bbV^\top \bbx=\diag\big(\tbh\big)\tbx=\tbh\circ \tbx.
\end{equation}
Apparently, the graph filter $\bbH$ is diagonalized by the graph's spectral basis $\bbV$. As a result $\tby$ is given by the element-wise product $(\circ)$ of $\tbx$ and the filter's frequency response $\tbh:=\bbPsi_L\bbh$, analogous to the convolution theorem for temporal signals. 


\subsection{Problem formulation}\label{ssec:prob_statement}


Suppose we observe $P$ graph signals that we collect in a matrix $\bbY=[\bby_1,\ldots,\bby_P] \in \reals^{N \times P}$. For given shift operator $\bbS$ and filter order $L$, observations adhere to the (graph convolutional) diffusion model $\bbY = \bbH \bbX$ in \eqref{eq:graph_filter_model}, where $\bbX = [\bbx_1,\ldots,\bbx_P]\in \reals^{N \times P}$ is sparse having at most $S\ll N$ non-zero entries per column. We do not require that the columns of $\bbX$ share a common support. The pragmatic setting whereby observations are corrupted by additive noise will be considered in Section \ref{ssec:stable_recovery}. 

The goal is to perform blind deconvolution on the graph $\ccalG$, which amounts to estimating sparse $\bbX$ and the filter coefficients $\bbh$ up to scaling and (possibly) permutation ambiguities~\cite[Sec. IV-A]{chang2018eusipco}. All we are given is $\bbY$, a forward model in terms of the parameterized filter family in \eqref{eq:graph_filter_def}, and a structural assumption on $\bbX$. Sparsity is well motivated when the signals in $\bbY$ represent diffused versions of a \emph{few} localized sources in $\ccalG$, here indexed by $\ccalS:=\textrm{supp}(\bbX)=\{(i,j) \mid X_{ij} \neq 0 \}$. From this vantage point, the blind deconvolution task can be viewed as a source localization one, where recovering $\bbX$ is of primary interest and $\bbh$ becomes a nuisance parameter. In any case, the non-sparse formulation is ill-posed, since the number of unknowns $NP + L$ in $\{\bbX,\bbh\}$ exceeds the $NP$  observations  in $\bbY$. 

All in all, using \eqref{eq:freq_response} the diffused source localization task can be stated as a non-convex feasibility problem of the form
\begin{equation} \label{eq:blind_feasibility}
\text{find } \{ \bbX,\bbh \} \:\: \text{s. to }\: \bbY = \bbV\diag\big(\bbPsi_L\bbh\big)\bbV^\top\bbX, \: \| \bbX \|_{0} \leq PS,
\end{equation}
where the $\ell_0$-(pseudo) norm $\| \bbX \|_{0}:=|\textrm{supp}(\bbX)|$ counts the non-zero entries in $\bbX$. In other words, the goal is to find the solution to a system of bilinear equations subject to a sparsity constraint in $\bbX$; a hard problem due to the cardinality function as well as the bilinear equality constraints. To deal with the latter, building on~\cite{wang2016blind} we will henceforth assume that the filter $\bbH$ is invertible. 

\section{Convex Relaxation for Invertible Graph Filters}\label{S:blind_ID}


In this section we propose a convex relaxation of \eqref{eq:blind_feasibility}, which is feasible under the additional mild assumption of diffusion filter invertibility. We wrap up with a brief discussion about algorithms to solve the resulting linear program. 


\subsection{A linear programming problem reformulation}\label{ssec:lp}


From the filter's input-output relationship in \eqref{eq:freq_response}, it follows that $\bbH$ will be invertible if $\tbh$ does not vanish at any of the graph frequencies $\{\lambda_i\}_{i=1}^N$. That is, no frequency component of the input should be completely annihilated by the filter. We make the following assumption on the signal model.
\begin{myassumption}[Invertible graph filter]
\normalfont
\label{assumption_gf}
Recall the observation model $\bbY=\bbH\bbX$ in Section \ref{ssec:model_diffusion}, where $\bbH=\sum_{l=0}^{L-1}h_l \bbS^l$ is a graph filter. We assume $\bbH$ is invertible, meaning $\tilde{h}_i=\sum_{l=0}^{L-1} h_l \lambda_{i}^l \neq 0$ holds for all  $i=1,\ldots,N$.   
\end{myassumption}
The inverse operator $\bbG := \bbH^{-1}$ is also a graph filter on $\ccalG$, which can be uniquely represented as a polynomial in the shift $\bbS$ of degree at most $N-1$ \cite[Theorem 4]{sandryhaila2013discrete}. Specifically, let $\bbg \in \reals^{N}$ be the vector of inverse-filter coefficients, i.e., $\bbG= \sum_{l=0}^{N-1} g_l \bbS^l$. Then one can equivalently rewrite the observation model $\bbY=\bbH\bbX$ as
\begin{equation} \label{eq:Filter_G}
		\bbX  = \bbG \bbY  = \bbV \text{diag}(\tbg) \bbV^\top \bbY,
\end{equation}
where $\tilde{\bbg} := \bbPsi_N \bbg \in \reals^N$ is the inverse filter's frequency response and $\bbPsi_N \in \reals^{N \times N}$ is a square Vandermonde matrix. Naturally, $\bbG=\bbH^{-1}$ implies the condition $\tbg\circ\tbh =\mathbf{1}_N$ on the frequency responses, where $\mathbf{1}_N$ denotes the $N \times 1$ vector of all ones. 

The fundamental implication of Assumption \ref{assumption_gf} is that, leveraging \eqref{eq:Filter_G}, one can recast \eqref{eq:blind_feasibility} as a \emph{linear} inverse problem
\begin{equation} \label{eq:opt_blind_invertible}
 \min_{\{\bbX,\bbg \}}
\:\| \bbX \|_{0},\:\: \text{s. to }\:
\bbX = \bbV \text{diag}(\bbPsi_N\bbg) \bbV^\top \bbY,\:\bbX\neq\mathbf{0}_{N\times P}.
\end{equation}
This approach to handle the bilinear equations \eqref{eq:blind_feasibility} is markedly different from the matrix lifting technique in~\cite{segarra2017blind}. We were inspired by the blind deconvolution method in~\cite{wang2016blind}, but in the GSP setting dealt with here the convolution kernel $\bbH$ is not circulant, and so $\bbV$ is no longer the discrete Fourier transform (DFT) matrix.

The $\ell_0$ pseudo-norm in \eqref{eq:opt_blind_invertible} renders the problem NP-hard to optimize. Convex-relaxation
approaches to tackle sparse recovery problems have enjoyed
remarkable success, since they often entail no loss of optimality. An important contribution of this work is to establish this holds for \eqref{eq:opt_blind_invertible} as well; see Section \ref{S:amb_uni}. Accordingly, we instead: (i) seek to minimize the $\ell_1$-norm convex surrogate of the cardinality function, that is $\| \bbX \|_{1,1} = \sum_{i,j}| X_{ij}|$; and (ii) express the filter in the graph spectral domain as in \eqref{eq:Filter_G}. This way, we arrive at the convex cost function 
\begin{align*}
	\| \bbX \|_{1,1}  = {} & \| \bbG \bbY \|_{1,1}\\
 = {} & \| \bbV \text{diag}(\tilde{\bbg}) \bbV^\top \bbY \|_{1,1}\\
	= {} &\| (\bbY^\top\bbV \odot \bbV) \tilde{\bbg} \|_{1},
\end{align*}
where the last equality is obtained after vectorizing the norm's matrix argument, and $\odot$ denotes the Khatri-Rao (i.e., columnwise Kronecker) product. Our idea is to solve the convex $\ell_1$-synthesis problem (in this case a linear program), e.g.,~\cite{zhang2016one}, namely
\begin{equation} \label{eq:opt_prob_convex}
\hat{\tilde{\bbg}}= \argmin_{\tilde{\bbg} \in \reals^{N}}
\:\| (\bbY^\top\bbV\odot \bbV) \tilde{\bbg} \|_{1},\quad \text{s. to }\:
\mathbf{1}_N^\top \tilde{\bbg} = N.
\end{equation}
While the linear constraint in \eqref{eq:opt_prob_convex} avoids $\hat{\tilde{\bbg}} = \mathbf{0}_{N}$, it also serves to fix the scale of the solution. An insightful discussion on the role of the constraint will emerge as an upshot of the exact recovery guarantees derived in Section \ref{S_s:unique}.


\subsection{Algorithms}\label{ssec:algorithms}


From an algorithmic point of view, under the mild assumption that the diffusion filter is invertible (cf. Assumption \ref{assumption_gf}), one can readily use e.g., an off-the-shelf interior-point method or a specialized sparsity-minimization algorithm to solve the linear programming formulation \eqref{eq:opt_prob_convex} efficiently. Different from the solvers in~\cite{segarra2017blind,david_blind}, the aforementioned algorithmic alternatives are free of expensive singular-value decompositions per iteration. 
Regardless of the particular algorithm chosen, we have found that overall performance can be improved via the iteratively-reweighted $\ell_1$-norm minimization procedure tabulated under Algorithm \ref{alg:reweighted_cvx}; see also~\cite{candes2008enhancing} for a justification of such refinement that effectively zeroes out small residual entries in intermediate estimates of $\bbX$. In practice, a couple refinement iterations suffice so the additional overhead is minimal. 

Once the frequency response $\hat{\tbg}$ of the inverse filter is recovered, one can readily reconstruct the sparse sources via
\begin{equation}\label{eq:X_estimate}
    \hbX =  \text{unvec}\left((\bbY^\top\bbV \odot \bbV) \hat{\tbg}\right),
\end{equation}
where here the $\text{unvec}(\cdot)$ operator reshapes its vector argument to an $N\times P$ matrix. If so desired, one can likewise form the diffusion filter estimate $\hbH=\bbV\text{diag}(\mathbf{1}_N/\hat{\tbg})\bbV^{\top}$, where the division is to be conducted entrywise.

\begin{algorithm}[t]
	\caption{Iteratively-reweighted $\ell_1$ minimization for \eqref{eq:opt_prob_convex}}
	\label{alg:reweighted_cvx}
	\begin{algorithmic}[1]
		\STATE 	\textbf{Input: } Matrix $\bbY^\top\bbV \odot \bbV$, $\delta > 0$ and  $\epsilon > 0$.
		\STATE \textbf{Initialize} $t=0$, $\bbw^{(0)}=\mathbf{1}_{NP}$,  $\bbX^{(0)} = \mathbf{0}_{N\times p}$.
		\REPEAT
		\STATE Solve \begin{equation*} 
     \tilde{\bbg}^{(t+1)}  = \argmin_{\tilde{\bbg}} \left\|  \bbw^{(t)}\circ [(\bbY^\top\bbV \odot \bbV) \tilde{\bbg}] \right\|_{1},\quad \text{s. to }\:\mathbf{1}_{N}^\top\tilde{\bbg} = N. 
  \end{equation*}
		\STATE Form $\textrm{vec}(\bbX^{(t+1)}) = (\bbY^\top\bbV \odot \bbV) \tilde{\bbg}^{(t+1)}$.\\
            \STATE Update $w_{i}^{(t+1)}=\frac{1}{[\textrm{vec}(\bbX^{(t+1)})]_i+\delta}, \:\: i=1,2,\ldots,NP.$ \\
		\STATE $t \gets t+1$.\\
		\UNTIL $\|\bbX^{(t)} - \bbX^{(t-1)}\|_{2} \le \epsilon$.\\
		\RETURN $\hat{\tilde{\bbg}}:= \tilde{\bbg}^{(t)}$ and $\hat{\bbX}:=\bbX^{(t)}$.
	\end{algorithmic}
\end{algorithm}


In the next section we will discuss the exact recovery condition for \eqref{eq:opt_prob_convex} and its robustness to noise corrupted observation.


\section{Recovery Guarantees}\label{S:amb_uni}

Here we conduct a theoretical analysis of the proposed convex estimator in \eqref{eq:opt_prob_convex}, which is a relaxation of the blind deconvolution problem on the graph $\ccalG$. We first derive exact recovery guarantees, which hold with high probability under a Bernoulli-Gaussian model for the sparse inputs. A stable recovery result is then established, ensuring the estimation error on the inverse filter's frequency response can be kept in check when the observations are corrupted by a small amount of additive noise.

Because of its analytical tractability, the Bernoulli-Gaussian model is widely adopted to describe and generate random sparse matrices such as the unknown sources $\bbX\in\reals^{N\times P}$ in Section \ref{ssec:prob_statement}. We henceforth adopt the following model specification that is consistent with the definition in \cite{li2015unified}.
\begin{mydefinition}[Bernoulli-Gaussian model] \normalfont 
\label{mydef_BGmodel}
We say a random matrix $\bbX\in\reals^{N\times P}$ adheres to a Bernoulli-Gaussian model with parameter $\theta\in (0,1)$, if its entries are $X_{ip}  = \Omega_{ip}\gamma_{ip}/\sqrt\theta$, where $\Omega_{ip}  \sim\text{Bernoulli}(\theta)$ and $\gamma_{ip} \sim \text{Normal}(0,1)$ are i.i.d. for all $i,p$.
\end{mydefinition}
In the context of Definition \ref{mydef_BGmodel} we say that the matrix entries $X_{ip}$ are Bernoulli-Gaussian random variables, with $\E{X_{ip}}=0$ and $\var{X_{ip}}=1$. Apparently, the model parameter $\theta$ offers a handle on the sparsity level of $\bbX$, while entries in $\textrm{supp}(\bbX)$ are drawn from a standard Normal distribution.  
%
%
%

\subsection{Exact recovery conditions} \label{S_s:unique}

Suppose that \eqref{eq:opt_blind_invertible} is identifiable (see e.g.,~\cite[Remark 1]{chang2018eusipco} for sufficient conditions under the Bernoulli-Gaussian model), and let $\{\bbX_0,\tilde{\bbg}_0 \}$ be the solution. For the ensuing discussion and to state our exact recovery result, some preliminary notation is in order. We compactly denote polynomials of the given graph-shift operator $\bbS=\bbV\bbLambda\bbV^\top$ as $\ccalP(\tbh):=\bbV\diag(\tbh)\bbV^\top$, where $\tbh$ is the corresponding filter's frequency response. So given observations $\bbY = \ccalP(\tbh_0)\bbX_0\in\reals^{N\times P}$ (implying $\bbX_0 = \ccalP(\tbg_0)\bbY$ under Assumption \ref{assumption_gf}, for $\tbg_0\circ\tbh_0 =\mathbf{1}_N$), we will study the following problem [cf. \eqref{eq:opt_prob_convex}]  
\begin{equation} \label{pb_1}
\hat{\tbg}  = \text{arg}\min\limits_{\tbg} \| \mathcal{P}(\tbg)\bbY\|_{1,1}
,\quad \text{s. to }\:
\bbr^\top\tbg = c,
\end{equation}
where $\bbr\in \reals^N$ is a generic constraint vector, and $c = \bbr^\top \tbg_0\neq 0$ is a constant to control the (inherently ambiguous) scale of the solution $\hat{\tbg}$. In general we do not have any prior knowledge about the ground-truth filter, so we can just select $\bbr = \mathbf{1}_\text{N}$ as in \eqref{eq:opt_prob_convex}. Our main result will offer some insights on how the recovery performance is affected by the choice of $\bbr$. 

We let $\bbP_1^\perp := \bbI_N-\frac{\mathbf{1}_N\mathbf{1}_N^\top}{N}$ denote the projection operator onto the orthogonal complement of the subspace $\textrm{span}(\mathbf{1}_N)$. Finally, consider the matrix $\tbU := (\bbV \circ \bbV)\bbP_1^\perp\in\reals^{N\times N}$, which one can show has a maximum singular value $\sigma_{\max}(\tbU)\leq 1$. 
With these definitions in place, we have all the elements to state the main theorem in this section. The result provides sufficient conditions to guarantee exact recovery of the inverse filter $\tbg_0$, and hence the sparse sources $\bbX_0$ via \eqref{eq:X_estimate}, with high probability.
\begin{mytheorem}[Exact recovery]\label{theorem_1}
\normalfont Consider graph signal observations $\bbY = \ccalP(\tbh_0)\bbX_0\in\reals^{N\times P}$, where $\bbX_0$ adheres to the Bernoulli-Gaussian model with $\theta\in \left(0,0.324\right]$. Recall that under Assumption \ref{assumption_gf}, we can write $\bbX_0 = \ccalP(\tbg_0)\bbY$. Let $P \geq C'\sigma_m^{-2}\log{\frac{4}{\delta}}$, where $\sigma_m = \min(\sigma_1,\sigma_2,\sigma_3,\sigma_4)$ and $\sigma_1\in \left(0,\frac{\sqrt{\pi}\theta^{3/2}}{2}\right]$, $\sigma_2 \in \left(0,\frac{\sqrt{\pi}\theta}{2}\right]$, $\sigma_3>0$, $\sigma_4\in(0,1)$,   $\delta\in(0,1)$ are parameters, while $C'$ is a constant that does not depend on $P$, {$\sigma_m$}, or $\delta$. Then $\hat\tbg = \tilde{\bbg}_0$ is the unique solution to \eqref{pb_1} with probability at least $1-\delta$, if
\begin{equation} \label{eq:reco_gaur}
\left\|\mathbf{P}^\perp_1\diag(\bbr)\Tilde{\bbg}_0\right\|_2  \leq c d_0,
\end{equation}
where $d_0 := \frac{\sqrt{1 - \sigma_{\max}^2(\Tilde{\bbU})}\left[(1-\sigma_1)-2\theta(1+\sigma_2)\right](1-\sigma_4)}{(1+\sigma_3)\sqrt\theta}$. 
\end{mytheorem} 
\begin{proof}
 See\ref{appendix_theo_1}.   
\end{proof}
Notice that when $\bbr = \mathbf{1}_\text{N}$ as in \eqref{eq:opt_prob_convex}, the sufficient recovery condition \eqref{eq:reco_gaur} simplifies to
\begin{align} \label{eq:reco_condi}
\left\|\mathbf{P}^\perp_1\Tilde{\bbg}_0\right\|_2  \leq N d_0. 
\end{align}
Condition \eqref{eq:reco_condi} essentially states that when the frequency response of the inverse filter $\tbg_0$ is closer to the all-ones vector $\mathbf{1}_N$, meaning that $\mathbf{P}^\perp_1\tbg_0$ is smaller in magnitude, the inverse filter is easier to recover, as the right-hand-side of the bound \eqref{eq:reco_gaur} does not directly depend on $\tbg_0$. Furthermore, the influence of prior knowledge of the true filter $\tbh_0$ on the recovery performance can be illustrated through a specific and straightforward case where $\bbr = \tbh_0$ in \eqref{eq:reco_gaur}. In this idealised scenario, since $\|\mathbf{P}^\perp_1\diag(\tbh_0)\tbg_0\|_2 = \|\mathbf{P}^\perp_1\mathbf{1}_N\|_2 = 0$, the left-hand-side of \eqref{eq:reco_gaur} will certainly be less than $c d_0$. 
Although this is a trivial case, it suggests that when $\bbr$ is closer to $\tbh_0$ (meaning that $\diag(\bbr)\tbg_0=\bbr\circ\tbg_0$ is closer to $\mathbf{1}_N$), the recovery performance of the proposed convex relaxation \eqref{pb_1} will improve.

Additionally, increasing $\theta$ within the feasible range $\theta\in \left(0, 0.324\right]$ will decrease $d_0$, making the recovery of the inverse filter more challenging. Conversely, when $\|\mathbf{P}^\perp_1\Tilde{\bbg}_0\|_2$ is smaller, a lower $d_0$ can be tolerated, allowing for a denser input signal (higher $\theta$). Moreover, in this last case the feasible parameters $\{\sigma_1, \sigma_2, \sigma_3, \sigma_4\}$ (and hence $\sigma_m$) can be larger, resulting in a lower required number of observations $P$. Naturally, increasing the recovery probability $1-\delta$ (provided that the sufficient condition is satisfied) necessitates a larger sample size; see the $\log(4/\delta)$ scaling.

\subsection{Stable recovery from noisy data}\label{ssec:stable_recovery}

Suppose we are now given $P$ noise-corrupted graph signal observations $\bbY = \ccalP(\tbh_0)\bbX_0+\bbN$, where $\bbN\in\reals^{N\times P}$ is an additive noise matrix. Since the filter is invertible, we can write $\ccalP(\tbg_0)\bbY = \bbX_0 + \ccalP(\tbg_0)\bbN$. It will be convenient to split the effective noise $\ccalP(\tbg_0)\bbN$ into matrix components respectively corrupting entries in $\ccalS:=\textrm{supp}(\bbX_0)$ and its complement $\ccalS^C$. The latter will be denoted by $\bbN^{(C)}:=[\ccalP(\tbg_0)\bbN]_{\ccalS^C} = \ccalP(\tbg_0)\bbN - [\ccalP(\tbg_0)\bbN]_{\ccalS}$, where matrix $[\bbM]_{\ccalA}$ has entries $M_{ij}$ if $(i,j)\in \ccalA$ and $0$ otherwise.

We can establish the following error bound for the solution of \eqref{pb_1}, which holds under the same conditions of Theorem \ref{theorem_1} and asserts that the inverse filter recovery error will be stable.
\begin{mytheorem}[Stable recovery] \label{theorem_2}
\normalfont
Consider graph signal observations $\bbY = \ccalP(\tbh_0)\bbX_0+\bbN$, where $\bbN\in\reals^{N\times P}$ is an additive noise matrix.
Let $\bbd: = \bbP_1^\perp \diag(\bbr)\Tilde{\bbg}_0$ and assume the conditions in Theorem \ref{theorem_1} are satisfied.  
Then the estimation error associated to the solution $\hat\tbg$ of problem \eqref{pb_1} 
is bounded by
\begin{equation} \label{eq_thm2}
    \left\|\hat\tbg - \tbg_0\right\|_l \leq\frac{2\left\|\diag(\tbg_0) \left(\bbI_N - \mathbf{1}_N\frac{(\bbr\circ\tbg_0)^\top}{c}\right)\right\|_{l\rightarrow 2}\|\bbN^{(C)}\|_{1,1}}{\sqrt{\frac{2}{\pi}}P Q - d_0  \|\bbN^{(C)}\|_{1,1}  - \|[\bbN^{(C)}]^\top\bbV\odot\bbV\|_{1\rightarrow 2}},
\end{equation}
where $Q:=\frac{(1+\sigma_3)\sqrt{\theta}}{c}\left[\sqrt{c^2 d_0^2-(1-\sigma_5)^2\|\bbd\|_2^2}-\sigma_5 \|\bbd\|_2\right]$, for some $\sigma_5\in[0,1]$, and $\|\cdot\|_l$ stands for the $\ell_1$ and $\ell_2$ norms when $l=1,2$, respectively.
\end{mytheorem}
\begin{proof}
See \ref{appendix_theo_2}.   
\end{proof}
Notice first that $Q\geq 0$ when \eqref{eq:reco_gaur} holds. Importantly, the denominator in the right-hand-side of \eqref{eq_thm2} should be non-negative to obtain feasible upper bound. This effectively imposes a constraint on the magnitude of the noise component $\bbN^{(C)}$ in $\ccalS^C$, which should satisfy 
\begin{equation}\label{eq_thm2_2}
    \|\bbN^{(C)}\|_F \leq \frac{\sqrt{\frac{2}{\pi}}P Q}{d_0\|\bar{\bbN}^{(C)}\|_{1,1} + \|[\bar{\bbN}^{(C)}]^\top\bbV\odot\bbV\|_{1\rightarrow 2}},
\end{equation}
where $\bar{\bbN}^{(C)} := \bbN^{(C)}/\|\bbN^{(C)}\|_F$. The right-hand-side of \eqref{eq_thm2_2} provides an upper bound to the strength of the noise that is tolerable. Once more, in favorable settings where $\|\bbd\|_2$ is small, e.g., if $\tbg_0$ is closer to the all-ones vector $\mathbf{1}_N$, we will have a larger upper bound in \eqref{eq_thm2_2} because $Q$ will be larger. Similarly for large $d_0$, for instance in sparse settings where $\theta$ is small. Either way, we have $Q\approx (1+\sigma_3)\sqrt{\theta} d_0$ and the noise condition simplifies to $\|\bbN^{(C)}\|_{1,1}\leq \sqrt{\frac{2 \theta}{\pi}} (1+\sigma_3)P$. 


\section{Numerical Results}\label{S:simulation}

We carry out numerical experiments to assess the performance of the proposed convex relaxation in a variety of settings. To this end, we run 
Algorithm \ref{alg:reweighted_cvx} and solve the per-iteration sparse recovery problems using CVX~\cite{grant2008cvx}. We first rely on synthetic data to simulate various controlled settings, which allow us to verify some of the conclusions drawn from our theoretical analysis in Section \ref{S:amb_uni}. We then compare against the matrix lifting approach in~\cite{segarra2017blind} with a focus on the sparsity levels and filter orders that lead to successful recovery. Finally, we conduct a network source localization experiment using real social network data, and include the non-convex $\ell_1$ recovery algorithm in~\cite{pena2016source} as an additional baseline. 

\begin{figure*}[ht]  
	\begin{minipage}[b]{0.24\textwidth}
		\centering
		\includegraphics[width=1\linewidth]{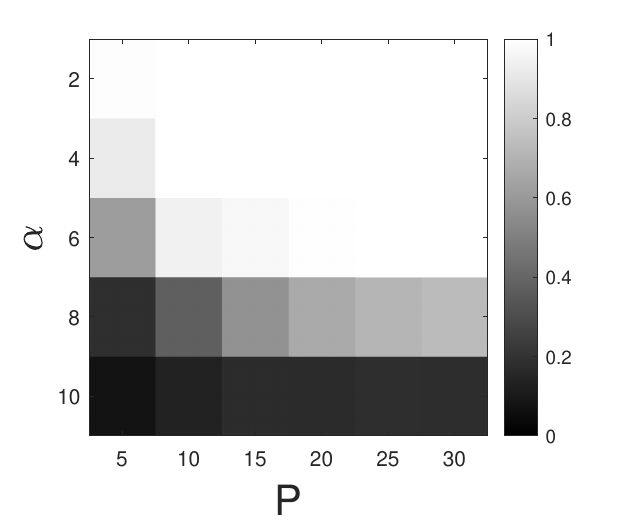}
		\centerline{(a)}\medskip
	\end{minipage}
	\hfill
	\begin{minipage}[b]{0.24\textwidth} 
		\centering
		\includegraphics[width=1\linewidth]{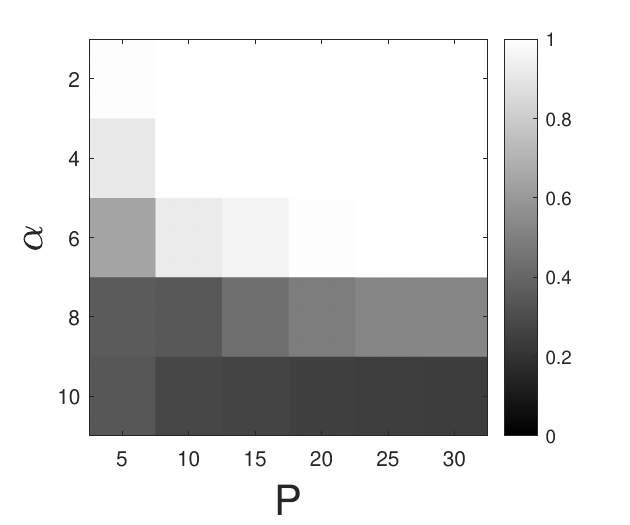}
		\centerline{(b)}\medskip
	\end{minipage}
	\hfill
	\begin{minipage}[b]{0.24\textwidth} 
		\centering
		\includegraphics[width=1\linewidth]{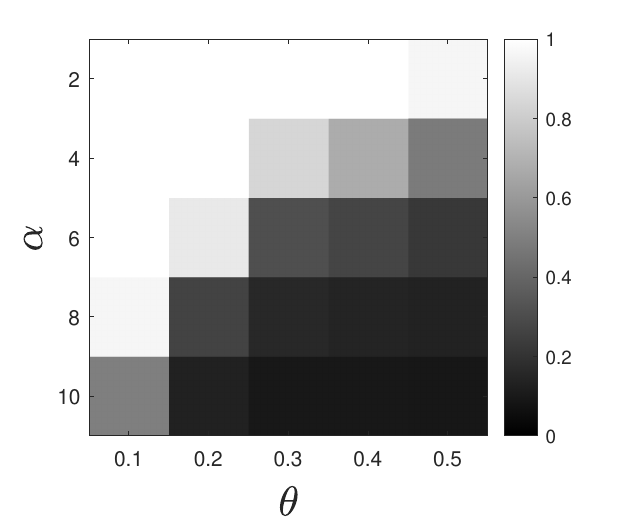}
		\centerline{(c)}\medskip
	\end{minipage}
	\hfill
	\begin{minipage}[b]{0.24\textwidth} 
		\centering
		\includegraphics[width=1\linewidth]{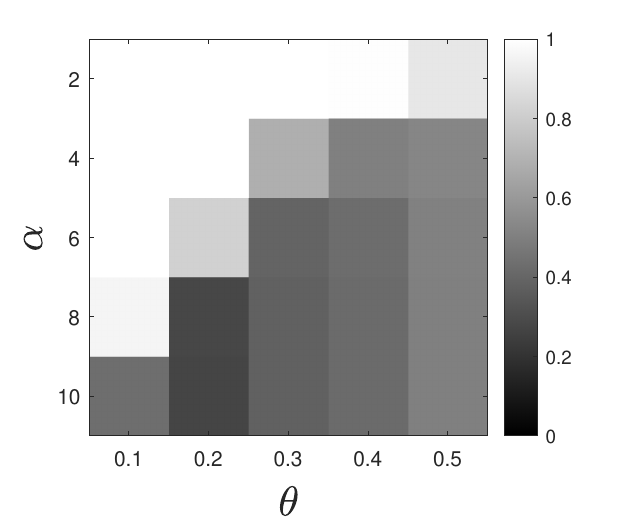}
		\centerline{(d)}\medskip
	\end{minipage}
	\vskip\baselineskip
	\begin{minipage}[b]{0.24\textwidth}
		\centering
		\includegraphics[width=1\linewidth]{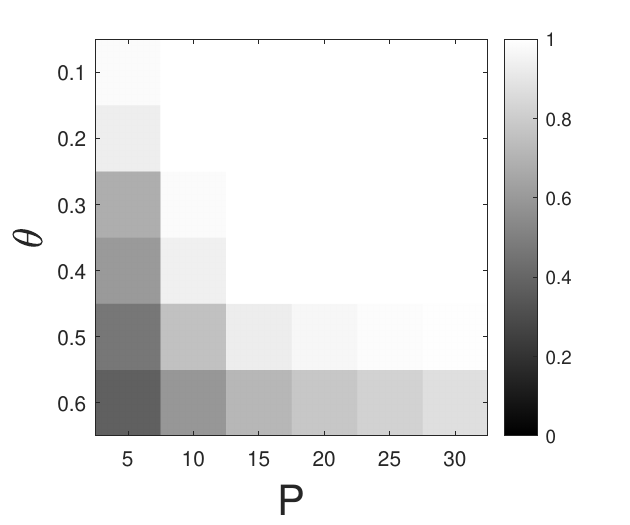}
		\centerline{(e)}\medskip
	\end{minipage}
	\hfill
	\begin{minipage}[b]{0.24\textwidth} 
		\centering
		\includegraphics[width=1\linewidth]{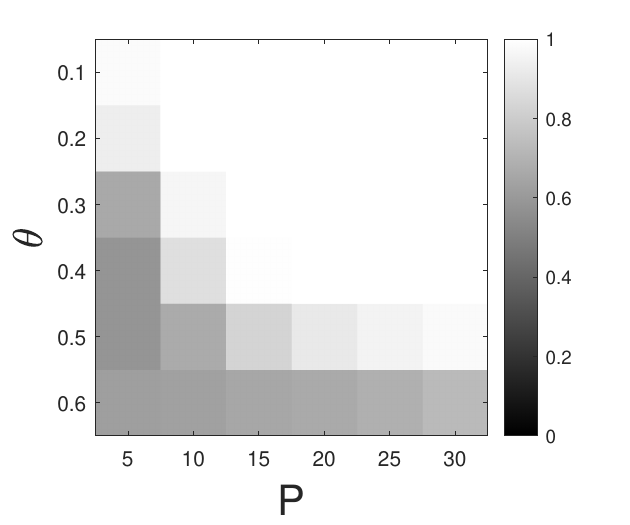}
		\centerline{(f)}\medskip
	\end{minipage}
	\hfill
	\begin{minipage}[b]{0.24\textwidth} 
		\centering
		\includegraphics[width=1\linewidth]{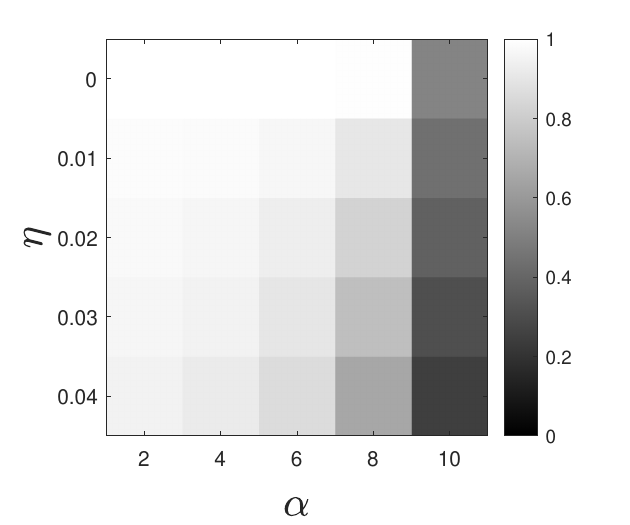}
		\centerline{(g)}\medskip
	\end{minipage}
	\hfill
	\begin{minipage}[b]{0.24\textwidth} 
		\centering
		\includegraphics[width=1\linewidth]{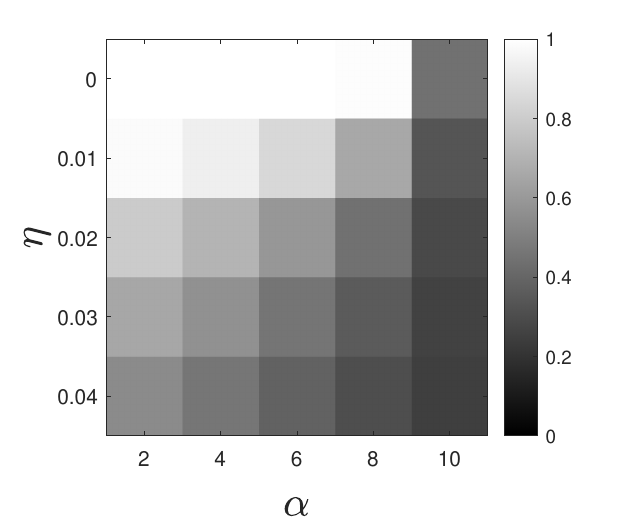}
		\centerline{(h)}\medskip
    \end{minipage}
 \caption{Recovery performance of Algorithm \ref{alg:reweighted_cvx} for problems involving Erdös-Renyi random graphs with $N = 20,\:p=0.4$. The first and third columns show $1 - \text{RE}_x$ for different settings (the whiter the better, meaning lower error), the second and fourth columns show the accuracy of support estimate for different settings (the whiter the better, meaning higher fraction of true sources recovered). For (a) and (b) the recovery performance of $\alpha$ v.s. $P$, with $\theta = 0.15$, is shown. For (c) and (d) $\alpha$ v.s. $\theta$ with $P = N = 20$. For (e) and (f) $\theta$ v.s. $P$, with $\alpha = 0.1N = 2$. For (g) and (h) we consider the noisy case, i.e. $\eta$ v.s. $\alpha$, with $\theta = 0.1, P = N = 20$. All of the results represent averages over 100 independent realizations. Apparently, the proposed approach exhibits satisfactory performance over a wide range of settings.} \label{f:figure_main}
\end{figure*}

\subsection{Key parameters affecting recovery performance}\label{ss:experiment_part1}

Here we consider Erdös-Renyi random graphs with $N = 20$ and edge connection probability $p=0.4$. For the GSO we adopt the degree-normalized adjacency matrix $\bbS = \bbD^{-\frac{1}{2}}\bbA \bbD^{-\frac{1}{2}}$, where $\bbD:=\textrm{diag}(\bbA \cdot \mathbf{1}_N)$. Notice that the eigenvalues of $\bbS$ are bounded, i.e., $\lambda_i\in(-1,1]$. 
We generate graph filters $\bbH$ such that their inverse filters have frequency responses $\tbg = \mathbf{1}_N + \alpha\bbP_1^\perp\bbb$, with perturbation $\bbb\in\reals^{N}$ drawn from $\textrm{Normal}(\mathbf{0}_N,\bbI_N)$ (multivariate standard Gaussian) and subsequently scaled such that $\|\bbP_1^\perp\bbb\|_2 = N$. Hence, the left-hand-side of \eqref{eq:reco_condi} becomes $\|\bbP_1^\perp\tbg_0\|_2 = \alpha\|\bbP_1^\perp\bbb\|_2 = \alpha N.$ and so \eqref{eq:reco_condi} will be satisfied whenever $\alpha\leq d_0$. The matrix $\bbX$ of sparse sources is drawn according to the Bernoulli-Gaussian model in Definition \ref{mydef_BGmodel}, and we will control the sparsity parameter $\theta$. Besides, in order to verify the stability properties conveyed by the error bound \eqref{eq_thm2}, we also consider noise corrupted observations $\bbY = \bbH\bbX+\eta\bbN$, where $\bbN$ has i.i.d. standard Gaussian entries. 

Given this simulation setup, we conducted four experiments to further examine the exact and stable recovery properties of \eqref{eq:opt_prob_convex}. The results are shown in Figure \ref{f:figure_main}, which depicts different figures of merit as a function of: (i) $\alpha$ vs $P$, (ii) $\alpha$ vs $\theta$, (iii) $\theta$ vs $P$, (iv) $\eta$ vs $\alpha$. As figures of merit we consider the root mean square error (RMSE) of the estimated sources $\hat\bbX$ as well as the support recovery accuracy, i.e., 
\begin{equation*}
\text{RE}_x:=\frac{\|\hat\bbX - \bbX_0\|_F}{\|\bbX_0\|_F},\quad \text{ACC}_x:=\frac{|\text{supp}_\kappa(\hat\bbX) \cap\text{supp}_\kappa(\bbX_0)|}{|\text{supp}_\kappa(\bbX_0)|},    
\end{equation*}
where $\text{supp}_\kappa(\cdot)$ is the support function with threshold $\kappa$. For these tests, we use $\kappa = 0.1$.

In Figure \ref{f:figure_main}, the first and third columns depict $1-\text{RE}_x$ and the second and fourth columns show $\text{ACC}_x$. In all cases, lighter-colored pixels indicate better performance. In the top row, we examine the relation of $\alpha = \|\bbP_1^\perp\tbg_0\|/N$ v.s. $P$ in (a)-(b) and $\alpha$ v.s. $\theta$ in (c)-(d). Our previous discussion for Theorem \ref{theorem_1} and the recovery condition \eqref{eq:reco_condi} is corroborated by these results which show a smaller $\alpha$ would both reduce the required sample size $P$ and tolerate denser sources, i.e., a larger $\theta$. From Figures \ref{f:figure_main} (e)-(f), it can be seen that while $\alpha = 2$ is fixed, which implies the lower bound of $d_0$ is unchanged, source signals generated with a larger $\theta$ require a larger sample size $P$ for successful recovery. And this observation is also consistent with Theorem \ref{theorem_1}. Notice that while our sufficient conditions hold for $\theta\in(0,0.324]$, in practice successful recovery is possible for larger $\theta$ provided that e.g., $P$ is large enough or $\alpha$ is sufficiently small. 

Moving on to the noisy setting, from Figures \ref{f:figure_main} (g)-(h) it can be seen that a smaller $\alpha$ would result in higher robustness to noise. This finding is also consistent with Theorem \ref{theorem_2}, because $\|\bbd\|_2 = \alpha$ when we use the all-ones vector in the constraint, i.e., $\bbr = \mathbf{1}_N$. All in all, as expected the numerical results in Figure \ref{f:figure_main} are in line with the theoretical guarantees in Theorems \ref{theorem_1} and \ref{theorem_2}.

\begin{figure*}[ht]  
	\begin{minipage}[b]{0.24\textwidth}
		\centering
		\includegraphics[width=1\linewidth]{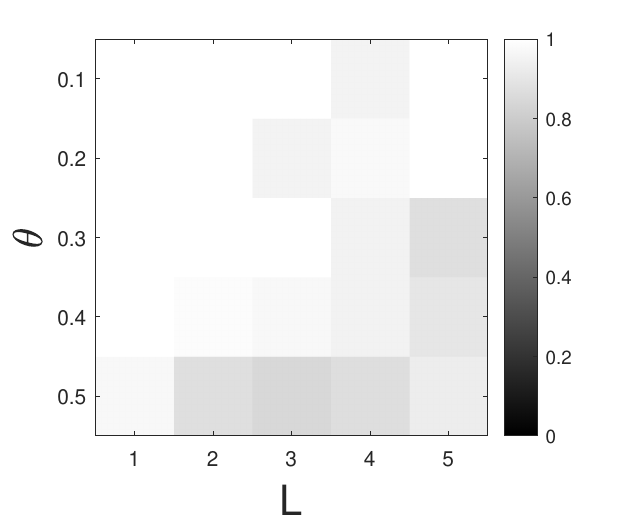}
		\centerline{(a)}\medskip
	\end{minipage}
	\hfill
	\begin{minipage}[b]{0.24\textwidth} 
		\centering
		\includegraphics[width=1\linewidth]{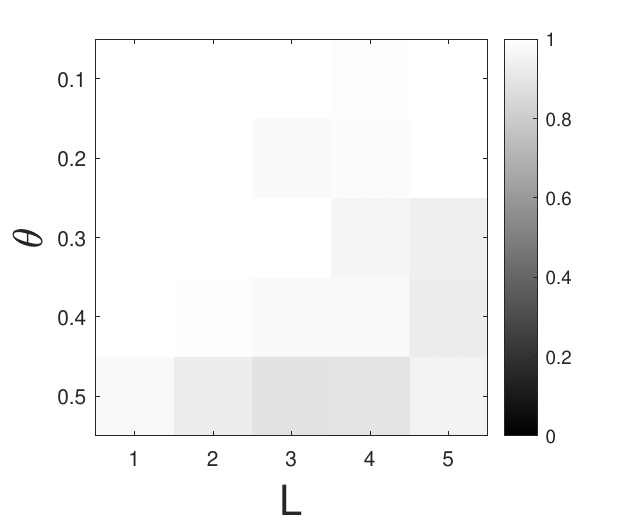}
		\centerline{(b)}\medskip
	\end{minipage}
	\hfill
	\begin{minipage}[b]{0.24\textwidth} 
		\centering
		\includegraphics[width=1\linewidth]{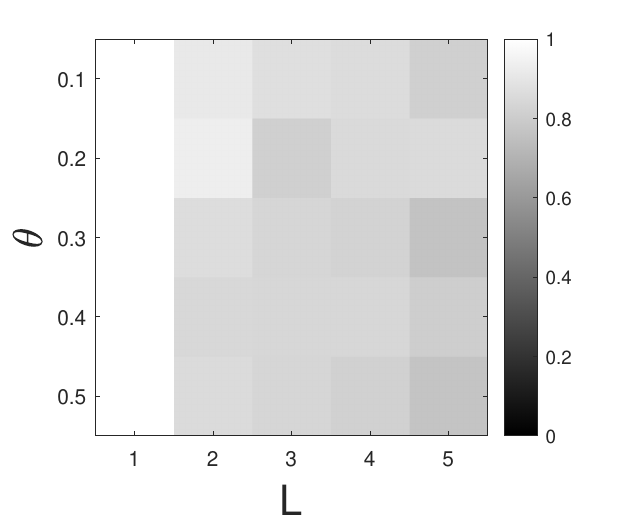}
		\centerline{(c)}\medskip
	\end{minipage}
	\hfill
	\begin{minipage}[b]{0.24\textwidth} 
		\centering
		\includegraphics[width=1\linewidth]{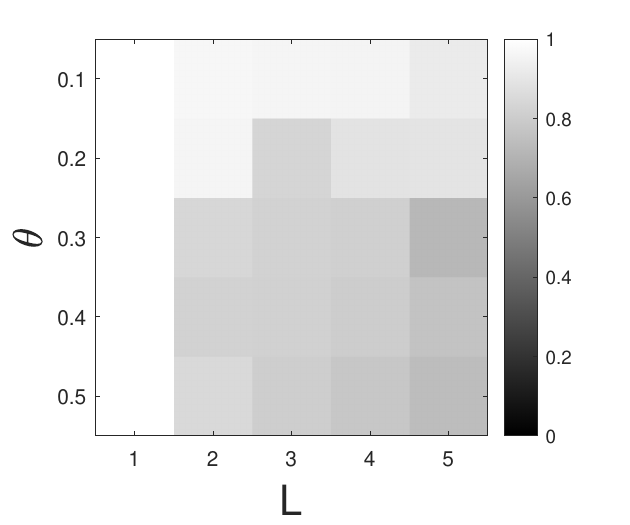}
		\centerline{(d)}\medskip
	\end{minipage}
 \vskip\baselineskip
    \begin{minipage}[b]{0.24\textwidth}
		\centering
		\includegraphics[width=1\linewidth]{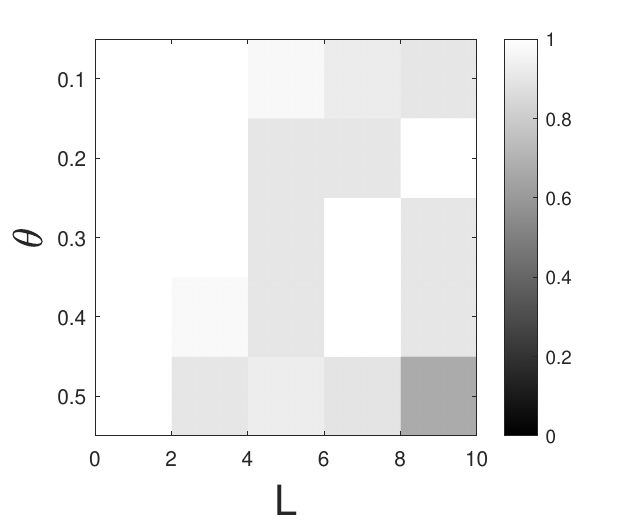}
		\centerline{(e)}\medskip
	\end{minipage}
	\hfill
	\begin{minipage}[b]{0.24\textwidth} 
		\centering
		\includegraphics[width=1\linewidth]{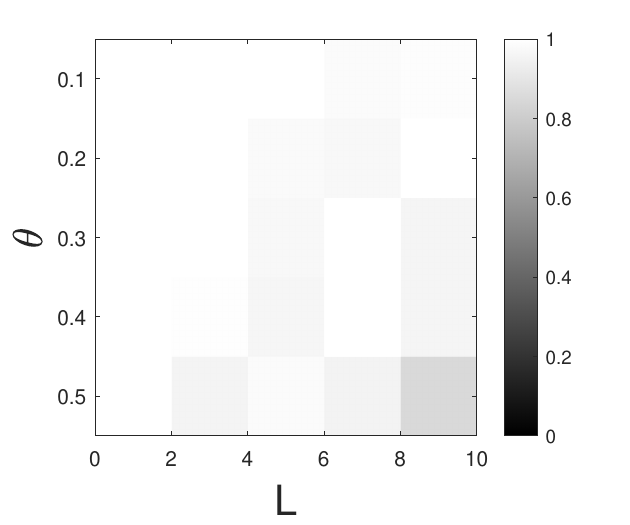}
		\centerline{(f)}\medskip
	\end{minipage}
	\hfill
	\begin{minipage}[b]{0.24\textwidth} 
		\centering
		\includegraphics[width=1\linewidth]{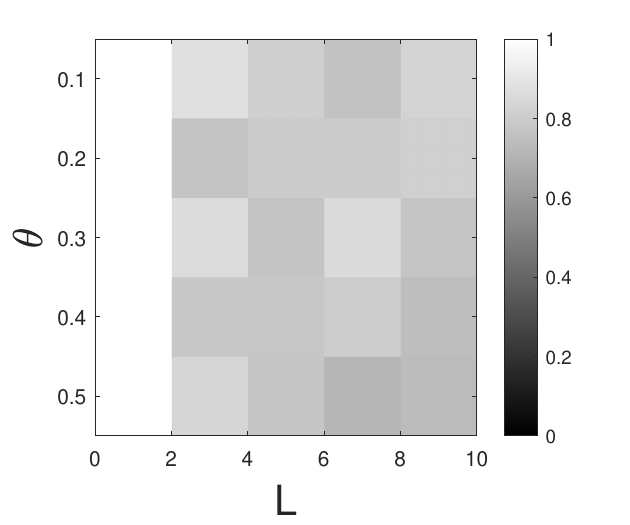}
		\centerline{(g)}\medskip
	\end{minipage}
	\hfill
	\begin{minipage}[b]{0.24\textwidth} 
		\centering
		\includegraphics[width=1\linewidth]{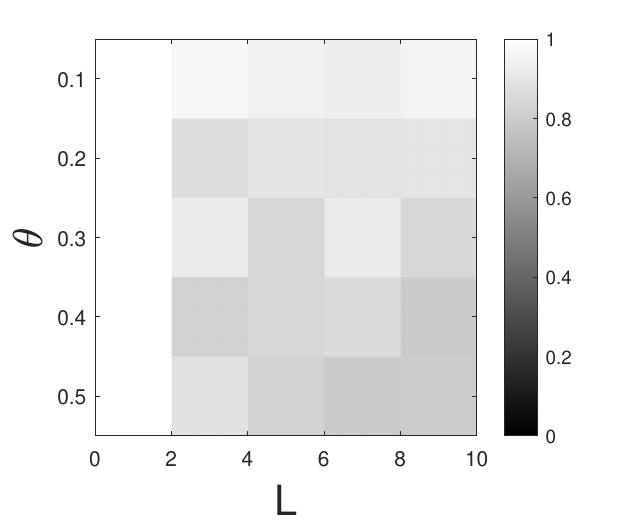}
		\centerline{(h)}\medskip
	\end{minipage}
\caption{Recovery performance comparing Algorithm \ref{alg:reweighted_cvx} with the matrix-lifting approach in~\cite{segarra2017blind}, for an ER graph with $N = P = 20,\: p=0.4$ and a filter perturbation magnitude of $\beta = 0.5$. Plots (a) and (b) show $1 - \text{RE}_x$ and accuracy of support estimate for the proposed estimator, while (c) and (d) show $1 - \text{RE}_x$ and accuracy of support estimate for the lifting approach. The same is shown in plots (e)-(h), but for the 66-node structural brain network studied in~\cite{hagmann2008mapping}, with $P=30$. All the results represent averages over 20 independent realizations. Algorithm \ref{alg:reweighted_cvx} uniformly outperforms  the matrix-lifting baseline, and is less sensitive to the filter order $L$.}\label{f:figure_L} 
\end{figure*} 

\subsection{Comparing with the matrix lifting approach in~\cite{segarra2017blind}} \label{ss:experiment_lifting}

Here we start by examining the sensitivity of the estimator \eqref{eq:opt_prob_convex} to the graph filter order $L$, which may impact $\|\bbd\|_2=\|\bbP_1^\perp\tbg_0\|_2$ and hence the satisfiability of \eqref{eq:reco_gaur}. To this end, let $\bbS = \bbD^{-\frac{1}{2}}\bbA \bbD^{-\frac{1}{2}}$ and consider graph filter coefficients of the form $\bbh = [1, h_1, ..., h_{L-1}]^\top$. The frequency response $\tbh = \bbPsi_L\bbh = \mathbf{1}_N + (\sum_{l=1}^{L-1}h_l)\bbe_1 + [\bbPsi_L]_{-1}\bbh_{-1}$, where $[\cdot]_{-1}$ zeroes out the first row of its matrix argument, i.e., $[\bbPsi_L]_{-1} = \diag(\mathbf{1}_N-\bbe_1)\bbPsi_L$ and likewise $\bbh_{-1}=\diag(\mathbf{1}_N-\bbe_1)\bbh$, where $\bbe_1=[1,0,\ldots,0]^\top\in \reals^N$. As for the normalized adjacency matrix used as GSO, all its eigenvalues are within $(-1,1)$ except for the first (biggest) eigenvalue $\lambda_1 = 1$, so the term $[\bbPsi_L]_{-1}\bbh_{-1}$ can be viewed as a small perturbation relative to $\mathbf{1}_N$. This suggests that a higher $L$ together with filter coefficients $\bbh_{-1}$ of larger magnitude would lead to greater deviations of $\tbh$ from $\mathbf{1}_N$. This in turn implies a bigger $\|\bbP_1^\perp\tbg_0\|_2$ that would make filter recovery harder. 
The preceding discussion suggests evaluating the recovery performance of Algorithm \ref{alg:reweighted_cvx} for graph filters with coefficients of the form $\bbh = \bbe_1 + \beta\bbh_{-1}$, with $\bbh_{-1}= [0, \bbb_{L-1}^\top]^\top$ and $\bbb_{L-1}\in\reals^{L-1}$ is drawn from $\textrm{Normal}(\mathbf{0}_{L-1},\bbI_{L-1})$. Parameter $\beta$ controls the filter perturbation, similar to $\alpha$ in Section \ref{ss:experiment_part1}. However, working with (and perturbing) $\bbh$ instead of the inverse filter $\tbg$ as in Section \ref{ss:experiment_part1}, directly allows us to examine the effect of $L$.

In this context, we compare the recovery performance of the proposed approach \eqref{eq:opt_prob_convex} against the matrix-lifting baseline in \cite{segarra2017blind} -- the latter is known to be more sensitive to $L$, especially for higher-order filters~\cite{chang2018eusipco}. We evaluate and report the same figures of merit $\text{RE}_x$ and $\text{ACC}_x$ used in the previous experiment, but now as a function of sparsity ($\theta$) and filter order ($L$). We fix the filter perturbation level to $\beta=0.5$. 
The numerical results are shown in Figures \ref{f:figure_L} (a)-(d) for a $N=20$-node Erdös-Renyi random graph with $p=0.4$, and in Figures \ref{f:figure_L} (e)-(h) for a $N=66$-node structural brain connectome from the study in~\cite{hagmann2008mapping}. The first two columns in Figure \ref{f:figure_L} depict $1-\textrm{RE}_x$ and $\text{ACC}_x$  for Algorithm \ref{alg:reweighted_cvx}, respectively; while the last two columns show the counterparts for the estimator in \cite{segarra2017blind}. Again, lighter-colored pixels are indicative of better performance. Results in Figure \ref{f:figure_L} clearly show that Algorithm \ref{alg:reweighted_cvx} uniformly outperforms the matrix-lifting baseline, and is markedly less (adversely) affected by large values of $L$ -- especially when the input sources are sufficiently sparse. The results are fairly consistent across the graphs tested here (compare the top and bottom rows in Figure \ref{f:figure_L}), but we note recovery appears to be more challenging when signals are diffused over the real brain connectome (even with 50\% more observations). In part, this could be explained by the fact that $\sigma_{\max}(\Tilde{\bbU}_{\textrm{ER}})=0.5054$ (averaged over 20 realizations of ER graphs with $N=20$ and $p=0.4$), while $\sigma_{\max}(\Tilde{\bbU}_{\textrm{brain}})=0.8769$. Hence, $d_0$ in \eqref{eq:reco_condi} will be smaller for the brain network -- everything else being equal.

\begin{figure*}[ht]
	\begin{minipage}[b]{0.24\textwidth}
		\centering
		\includegraphics[width=1\linewidth]{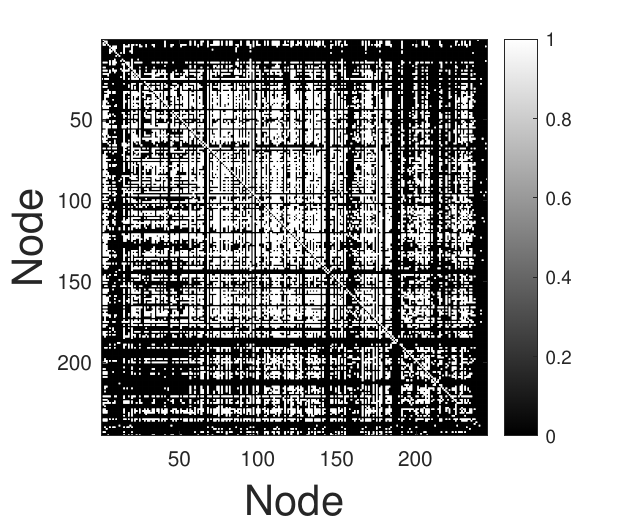}
		\centerline{(a)}\medskip
	\end{minipage}
	\hfill
	\begin{minipage}[b]{0.24\textwidth} 
		\centering
		\includegraphics[width=1\linewidth]{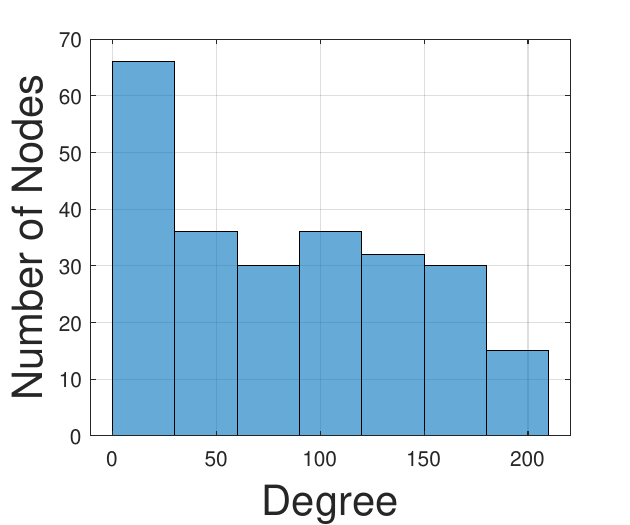}
		\centerline{(b)}\medskip
	\end{minipage}
	\hfill
	\begin{minipage}[b]{0.24\textwidth} 
		\centering
		\includegraphics[width=1\linewidth]{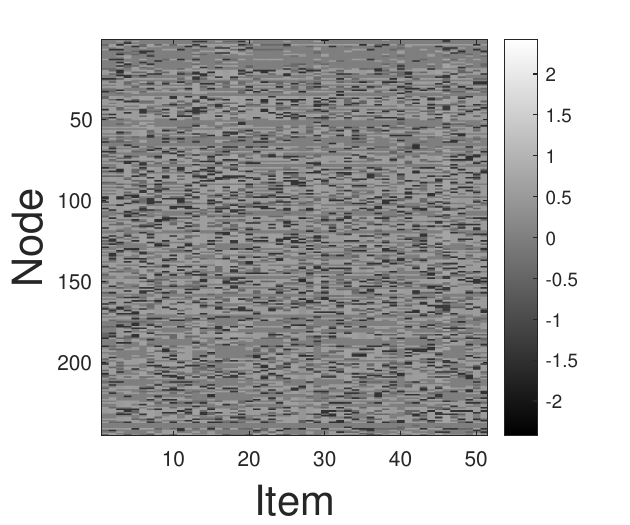}
		\centerline{(c)}\medskip
	\end{minipage}
	\hfill
	\begin{minipage}[b]{0.24\textwidth} 
		\centering
		\includegraphics[width=1\linewidth]{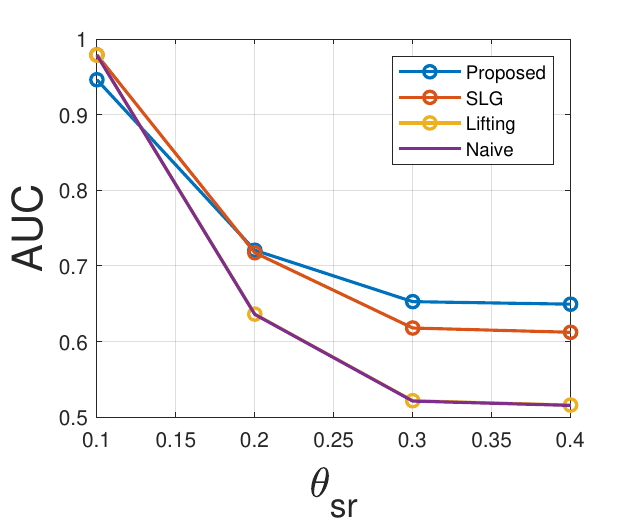}
		\centerline{(d)}\medskip
	\end{minipage}

\caption{Recovery performance on the Epinions dataset. (a) The adjacency matrix $\bbW$ of the sampled directed social network and the (b) the degree distribution of the symmetrized undirected graph with adjacency $\bbA=(\bbW+\bbW^\top)/2$. (c) The centered rating data $\bbY_\text{obs}$. (d) Source localization performance is quantified in terms of the AUC as a function source density level $\theta_\text{sr} = \{0.1,0.2,0.3,0.4\}$. We compare four different approaches: (i) the proposed estimator \eqref{eq:opt_prob_convex}; (ii) the source localization on graphs (SLG) algorithm in~\cite{pena2016source}, (iii) the matrix lifting-based approach from~\cite{segarra2017blind}, and (iv) the naive predictor whereby $\hbX=\bbY_\text{obs}$. We find that Algorithn \ref{alg:reweighted_cvx} offers more robust predictions across a broader range of input sparsity levels.}\label{f:figure_epinion}
\end{figure*}

\subsection{Experiments on real social network data}\label{ssec:real_data}

In this section we test the proposed approach on the Epinions dataset~\cite{hamedani2021trustrec}, a who-trusts-whom online social network that includes 132k users, 755k items (articles written by some of the users), 13M user-to-item ratings (1-5 scale, with timestamp), and the signed trust/distrust pairwise relations (717k for trust and 123k for distrust) between users. Leveraging these data, we want to tackle the following network source localization problem: given a connected social graph and timestamped item ratings generated by users in this social network, can we locate the subset of users who generate the earliest ratings, assuming that other users' rating might be impacted by those early ratings? 

More precisely, for a connected social network with $N$ user nodes and a set of $P$ items, we want to identify the $\theta_\text{sr}$ earliest ratings ($\theta_\text{sr}$ is a prescribed proportion of earliest ratings assigned to these items, say $\theta_\text{sr} = 10\%, 20\%, 30\%, 40\%$) from the observed ratings matrix $\bbY_\text{obs}\in\reals^{N\times P}$.  Note that the rating density of the whole dataset is fairly low, i.e., $0.0015\%$~\cite{hamedani2021trustrec}. As a result, for a sparse observation matrix $\bbY_\text{obs}$ the optimal solution of \eqref{eq:opt_prob_convex} would likely yield the trivial result $\hat{\tbg} = \mathbf{1}_N$, since this solution both satisfies the recovery condition \eqref{eq:reco_gaur} and is compatible with the sparsity requirement on the sources $\hat{\bbX} = \bbV\diag(\hat{\tbg})\bbV^\top\bbY_\text{obs} = \bbY_\text{obs}$. To obtain an interesting problem instance that is compatible with our setting, we generate a sub-dataset with higher rating density. To this end, we sample and pre-process the original data (details can be found in \ref{appendix_data}), and obtain a reduced dataset of $N=245$ users and $P = 51$ items with rating density $0.64$. All of the $N$ users were connected via a directed trust network $\ccalG$ with binary adjacency matrix $\bbW$, i.e., the link $W_{ij}$ from user $i$ to user $j$ indicates $i$ was trusted by $j$, and hence user $i$'s opinion would impact that of user $j$. The resulting adjacency matrix $\bbW$ is shown in Figure \ref{f:figure_epinion} (a) and the degree distribution of the symmetrized undirected graph with $\bbA=(\bbW+\bbW^\top)/2$ can be found in Figure \ref{f:figure_epinion} (b). Like in the previous experiments, we adopt $\bbS=\bbD^{-\frac{1}{2}}\bbA\bbD^{-\frac{1}{2}}$ as GSO. Following a data pre-processing step, the resulting centered ratings matrix $\bbY_\text{obs}$ with values in $[-2,2]$ is shown in Figure \ref{f:figure_epinion} (c). Given that the ground-truth sparsity of the source $\theta_\text{sr}$ is unknown, we examine different sparsity level assumptions on the input, namely, $\theta_\text{sr} \in \{0.1, 0.2,0.3,0.4\}$. Specifically, to populate the source signal $\bbX_\text{sr}$ we retain different proportions $\theta_\text{sr}$ of the earliest ratings per item $p=1,\ldots,P$. This way $\theta_\text{sr}$ is closely related to $S$ as defined in Section \ref{ssec:prob_statement}. 

To assess the source recovery performance, we take $\bbS$ and $\bbY_\text{obs}$ as inputs and compare three approaches: two methods for graph-aware blind deconvolution including the proposed estimator \eqref{eq:opt_prob_convex} and the lifting approach in \cite{segarra2017blind}, plus the non-convex $\ell_1$ recovery algorithm for source localization on graphs (SLG)~\cite{pena2016source}. Since the support of sources $\bbX_\text{sr}$ is a strict subset of the  support of $\bbY_\text{obs}$, we consider the area under the curve (AUC) of the predicted sources in $\hbX$ as figure of merit in this numerical test case. 
In addition to the aforementioned three methods, we also consider a naive baseline whereby $\hbX=\bbY_\text{obs}$. The resulting AUCs for different source signal sparsity levels $\theta_{sr}$ are shown in Figure \ref{f:figure_epinion} (d). It can be seen that Algorithm \ref{alg:reweighted_cvx} achieves the highest AUC for the denser settings $\theta_{sr} = 0.3, 0.4$; for $\theta_{sr} = 0.2$, the proposed approach and SLG attain a similar AUC (which is higher than the naive baseline); for $\theta_{sr} = 0.1$, the proposed method performs marginally worse than all of the other three predictors. The matrix lifting algorithm is only competitive when the sources are the sparsest. 

Overall, we find that the proposed estimator \eqref{eq:opt_prob_convex} offers more robust predictions across a broader range of input sparsity levels, a promising finding to support the prospect of solving real world network deconvolution problems. In all fairness though, the performance of none of the methods is stellar. 
But we note these real data are complex and we lack a ground truth for validation, since the chronological order alone does not imply causality. Secondly, even through there are a few earliest ratings that impact other ratings and hence should be reasonably viewed as sources, the observations could still be highly noisy. 


\section{Concluding Summary and Future Work}\label{S:conclusion}

We studied the problem of blind graph filter identification from multiple sparse inputs, which extends blind deconvolution of time (or spatial) domain signals to graphs. By introducing a mild assumption on invertibility of the graph filter, we obtained a computationally simpler convex relaxation for (diffused) source localization in the multi-signal case. In terms of theoretical analyses, we first derived sufficient conditions for exact recovery, which hold with high probability under a Bernoulli-Gaussian model for the sparse inputs. A stable recovery result is then established, ensuring the estimation error on the inverse filter's frequency response is manageable when the observations are corrupted by a small amount of noise.

Ongoing work includes additional analyses on the \emph{robustness} of the proposed approach to imperfections in the observed graph, as well as when measurements are collected only in a fraction of nodes. On the algorithmic side, developing an {\it online} network source localization method capable of processing streaming graph signal observations is also of interest.

\appendix
\section{Proof of Theorem \ref{theorem_1}} \label{appendix_theo_1}


Recall the notation introduced in Section \ref{S_s:unique}. The proof of Theorem \ref{theorem_1} begins by considering an equivalent problem to \eqref{pb_1}, obtained via an invertible change of variable $\bbw = \tbg \circ \tbh_0$, namely 
%
\begin{equation} \label{pb_2}
\hat{\bbw} = \argmin_{\bbw} \| \ccalP(\bbw) \bbX\|_{1,1},\quad \text{s. to }\: \bar{\bbr}^\top \bbw  = c,
\end{equation}
where $\bar{\bbr}^\top = \bbr^\top\text{diag}(\tbg_0)$. Note that the solution candidate $\hat{\bbw} = \mathbf{1}_N$ implies $\hat\tbg = \tbg_0$ in \eqref{pb_1}, so we let $c = \bar{\bbr}^\top \mathbf{1}_N$. Then we have the following proposition, which simply restates Theorem \ref{theorem_1} in terms of the equivalent problem \eqref{pb_2}.
\begin{myproposition}[Exact recovery for the equivalent problem]\label{proposition:1}
\normalfont Consider graph signal observations $\bbY = \ccalP(\tbh_0)\bbX_0\in\reals^{N\times P}$, where $\bbX_0$ adheres to the Bernoulli-Gaussian model with $\theta\in \left(0,0.324\right]$. Recall that under Assumption \ref{assumption_gf}, we can write $\bbX_0 = \ccalP(\tbg_0)\bbY$. Let $P \geq C'\sigma_m^{-2}\log{\frac{4}{\delta}}$, where $\sigma_m = \min(\sigma_1,\sigma_2,\sigma_3,\sigma_4)$ and $\sigma_1\in \left(0,\frac{\sqrt{\pi}\theta^{3/2}}{2}\right]$, $\sigma_2 \in \left(0,\frac{\sqrt{\pi}\theta}{2}\right]$, $\sigma_3>0$, $\sigma_4\in(0,1)$,   $\delta\in(0,1)$ are parameters, while $C'$ is a constant that does not depend on $P$, {$\sigma_m$}, or $\delta$. Then $\hbw = \mathbf{1}_N$ is the unique solution to \eqref{pb_2} with probability at least $1-\delta$, if
\begin{equation} \label{theorem0_1}
\left\|\mathbf{P}^\perp_1\bar{\bbr}\right\|_2  \leq c d_0,
\end{equation}
where $d_0 := \frac{\sqrt{1 - \sigma_{\max}^2(\Tilde{\bbU})}\left[(1-\sigma_1)-2\theta(1+\sigma_2)\right](1-\sigma_4)}{(1+\sigma_3)\sqrt\theta}$. 
\end{myproposition}
%
To establish Proposition \ref{proposition:1} (and thus Theorem \ref{theorem_1}), we will derive a concentration property for hollow matrices with an all-zero diagonal [cf. \eqref{proof_prop1_lowerbound2}] 
that follows from Proposition \ref{proposition:2}. 
\begin{myproposition}\label{proposition:2}
Consider vectors $\bbm_i\in\reals^N,\: i\in\{1,\ldots,N\}$, such that $[\bbm_i]_i = 0$. Suppose $\bbX\in\reals^{N\times P}$ is drawn form the Bernoulli-Gaussian model in in Definition \ref{mydef_BGmodel}, with $\theta\in \left(0,e^{-1}\right]$. 
Given parameters $\delta\in(0,1)$, $\sigma_1\in \left(0,\frac{\sqrt{\pi}\theta^{3/2}}{2}\right]$, and $\sigma_2 \in \left(0,\frac{\sqrt{\pi}\theta}{2}\right]$, let $P \geq \frac{C}{\min(\sigma_1^2,\sigma_2^2)}\log{\frac{4}{\delta}}$ for some constant $C$.
Then, for each $i\in\{1,\ldots,N\}$ we have
%
\begin{equation} \label{prop:hollow_vector}
    \begin{aligned}
    \textbf{(a)}&\:\Pr{\Bigl| \frac{1}{\beta_i P} \|\bbm^\top_i \bbX\|_1 -\|\bbm_i\|_2 \Bigr| \leq \sigma_1\| \bbm_i\|_2 } \geq 1-\delta\\  
    \textbf{(b)}&\:\Pr{\Bigl| \frac{1}{\beta_i \theta P} \|(\bbomega_i^\top)\circ (\bbm^\top_i \bbX)\|_1 - \| \bbm_i\|_2 \Bigr| \leq \sigma_2\| \bbm_i\|_2 } \geq 1-\delta,
    \end{aligned}
\end{equation}
%
where $\beta_i := \E{\|\bbm^\top_i \bbX\|_1}/\|\bbm_i\|_2$ and $\bbomega_i = [\Omega_{i1},\hdots,\Omega_{iP}]^\top$. 
\end{myproposition}
%

%
%
The idea behind (a) in \eqref{prop:hollow_vector} comes from \cite[Theorem 5.1]{matouvsek2008variants}, namely that the absolute value of linear combinations of i.i.d. Bernoulli-Gaussian $\{X_{ij}\}$ will concentrate to its expectation. For (b), let $\bbx_p$ be the $p$-th column of $\bbX$ and note that $\Omega_{ip}$ and $|\bbm_i^\top \bbx_p|$ are independent $\forall i,p$, when $\bbm_i$ is a hollow vector. This implies 
\begin{align*}
\E{\|(\bbomega_i^\top)\circ (\bbm^\top_i \bbX)\|_1}= {}&\E{\sum_{p=1}^P \Omega_{ip}|\bbm^\top_i \bbx_p|}\\
= {}&\mathbb{E}_{\bbX}\left[\sum_{p=1}^P \E{\Omega_{ip}}|\bbm^\top_i \bbx_p|\right]\\
= {}& \theta\mathbb{E}_{\bbX}\left[\sum_{p=1}^P|\bbm^\top_i \bbx_p|\right].
\end{align*}

To prove Proposition \ref{proposition:2}, we first establish preliminary Lemmata \ref{lemma:1} and \ref{lemma:2} that are similar to~\cite[Lemma 5.3]{matouvsek2008variants}, and then prove Lemmata \ref{lemma:exp_tail} and \ref{lemma:gaussian_tail} that follow ideas from~\cite[Proposition 5.2]{matouvsek2008variants}. Specifically, in the first step we show that for a hollow matrix $\bbM$ and Bernoulli-Gaussian distributed $\bbX$, the entries of both of $\bbM \bbX$ and $\bbOmega\circ(\bbM \bbX)$ have bounded expectation and variance (cf. Lemma \ref{lemma:1} and Lemma \ref{lemma:2}). Recalling $\ccalS:=\textrm{supp}(\bbX)=\textrm{supp}(\bbOmega)$, let $\bar\bbX :=\bar\bbM \bbX$ and $\bar\bbX^{(\ccalS)}:=\frac{1}{\theta}\bbOmega\circ(\bar\bbM \bbX)$, with $\bar\bbM := [\bbm_1/\|\bbm_1\|_2,\ldots,\bbm_N/\|\bbm_N\|_2]^\top\in\reals^{N\times N}$. Then in the second step, in Lemma \ref{lemma:exp_tail} we show that $|\bar{X}_{ip}|$ (and $|\bar{X}^{(\ccalS)}_{ip}|$) have uniform exponential tails, meaning that for some constant $b>0$ and all $t\geq 0$ we have $\Pr{|\bar{X}_{ip}|\geq t}\leq e^{-bt}$. Finally, Lemma \ref{lemma:gaussian_tail} asserts that the sum of the entries of both $\bar\bbX$ and $\bar\bbX^{(\ccalS)}$ will concentrate to their entries' mean.

\begin{mylemma} \label{lemma:1}
Consider vectors $\bar\bbm_i\in\reals^N\: i\in\{1,\ldots,N\}$, such that $[\bar{\bbm}_i]_i = 0$, $\|\bar\bbm_i\|_\infty = \alpha_i$, and $\|\bar\bbm_i\|_2 = 1$. Let $\bar\bbX^{(\ccalS)} = \frac{1}{\theta}\bbOmega\circ(\bar\bbM \bbX)\in\reals^{N\times P}$, where $\bbX\in\reals^{N\times P}$ is drawn form the Bernoulli-Gaussian model in Definition \ref{mydef_BGmodel}, with $\theta\in \left(0,e^{-1}\right]$, and $\bar\bbM = [\bar\bbm_1,\hdots,\bar\bbm_N]^\top\in\reals^{N\times N}$. Then $\E{|\bar{X}^{(\ccalS)}_{ip}|^2} = \frac{1}{\theta}$, and $\beta_0 (1- \sigma'_i) \leq \E{|\bar{X}^{(\ccalS)}_{ip}|} \leq\beta_0 ,\forall (i,p)$, where $\beta_0:=\E{|\gamma_{ip}|}=\sqrt{2/\pi}$ 
and
\begin{equation} \label{eq.lemma3.1}
    \sigma'_i = \left\{
    \begin{aligned}
        & \frac{\alpha_i^2}{\theta}, & \alpha_i\in (0,\sqrt\theta]\cap [N^{-1/2},1],\\
        & 1 - \sqrt{\theta}\alpha_i \left[1 + \frac{(1-\alpha_i^2)\theta^2}{2\alpha_i^2(1-\theta)}\right], & \alpha_i\in (\sqrt{\theta},1]\cap[N^{-1/2},1].
    \end{aligned}
    \right.
\end{equation}    
\end{mylemma} 
\begin{mylemma} \label{lemma:2}
Consider vectors $\bar\bbm_i\in\reals^N\: i\in\{1,\ldots,N\}$, such that $[\bar{\bbm}_i]_i = 0$, $\|\bar\bbm_i\|_\infty = \alpha_i$, and $\|\bar\bbm_i\|_2 = 1$. Let $\bar\bbX = \bar\bbM \bbX\in\reals^{N\times P}$, where $\bbX\in\reals^{N\times P}$ is drawn form the Bernoulli-Gaussian model in Definition \ref{mydef_BGmodel} with $\theta\in (0,e^{-1}]$, and $\bar\bbM = [\bar\bbm_1,\hdots,\bar\bbm_N]^\top\in\reals^{N\times N}$.
Then $\E{|\bar{X}_{ip}|^2} =1$ and $\beta_0 (1- \sigma_i')\leq \E{|\bar{X}_{ip}|} \leq\beta_0,\forall (i,p)$, where $\sigma_i'$ is given by \eqref{eq.lemma3.1}.
\end{mylemma}
   
\noindent\textbf{Proof of Lemma \ref{lemma:1}}. Notice that for any $(i,p)$, $\bar{M}_{ii} = 0$, so $\bar{X}^{(\ccalS)}_{ip} = \frac{1}{\theta}\Omega_{ip}\sum_{j = 1}^N\bar{M}_{ij}X_{jp} = \frac{1}{\theta}\Omega_{ip}\sum_{j\not = i}\bar{M}_{ij}\Omega_{jp}\gamma_{jp}\theta^{-1/2}$. By conditioning on $Z_{ip}=\sum_{j\not = i}\Omega_{jp}\bar{M}_{ij}^2=z$, then $\bar{X}_{ip} = \theta^{-1/2}\sum_{j\not = i}\bar{M}_{ij}\Omega_{jp}\gamma_{jp}$ is distributed as $\textrm{Normal}(0,\frac{z}{\theta})$. Also notice that $\theta \bar{X}^{(\ccalS)}_{ip} = \Omega_{ip}\bar{X}_{ip}$, where $\Omega_{ip}$ and $\bar{X}_{ip}$ are independent. Then $\E{|\theta \bar{X}^{(\ccalS)}_{ip}|\:\given Z_{ip}=z} = \theta\beta_0\sqrt{\frac{z}{\theta}}$, and thus $\E{|\bar{X}^{(\ccalS)}_{ip}|}= \theta^{-1/2}\beta_0\E{\sqrt{Z_{ip}}}$. Following ideas in the proof of \cite[Lemma 5.3]{matouvsek2008variants}, we can derive upper and lower bounds for  $\E{\sqrt{Z_{ip}}}$. 

For the upper bound, we have $\E{\sqrt{Z_{ip}}}\leq \sqrt{\E{Z_{ip}}} \leq \sqrt{\theta}$. For the lower bound, we consider two cases, $\alpha_i<\sqrt{\theta}$ and $\alpha_i\geq\sqrt{\theta}$. When $\alpha_i<\sqrt{\theta}$, we use the strategy in \cite[Lemma 5.3]{matouvsek2008variants}. Let $t= \frac{Z_{ip}}{\theta}-1\geq -1$, and recall that $\sqrt{1+t}\geq 1 + \frac{t}{2} - t^2$, hence 
\begin{align*}
\E{\sqrt{Z_{ip}}} = {} &\sqrt{\theta}\E{\sqrt{1+t}}\\
\geq {} & \sqrt{\theta}\left(1+\E{\frac{Z_{ip}}{\theta}-1} - \E{\left(\frac{Z_{ip}}{\theta}-1\right)^2}\right)\\ 
={} &\sqrt{\theta}\left(1-\theta^{-2}\var{Z_{ip}}\right).    
\end{align*}
Now $\var{Z_{ip}}=\sum_{j\not = i}\bar{M}_{ij}^4\var{\Omega_{jp}} \leq\sum_{j\not = i} \bar{M}_{ij}^4 \theta \leq \theta\alpha_i^2\sum_{j\not = i}\bar{M}_{ij}^2 = \theta\alpha_i^2$ (as $\bar\bbm_i$ has unit $\ell_2$ norm). Hence, $\E{\sqrt{Z_{ip}}}\geq \sqrt{\theta}(1-\frac{\alpha_i^2}{\theta}) \geq \sqrt{\theta}(1-\sigma'_i)$, where $\sigma'_i = \frac{\alpha_i^2}{\theta}$.  As a result, $\beta_0 (1 - \sigma'_i) \leq \E{|\bar{X}^{(\ccalS)}_{ip}|} \leq\beta_0,\forall (i,p)$, with $\sigma'_i = \alpha_i^2/\theta<1$.

When $\alpha_i\geq\sqrt{\theta}$, we let $k$ bet the index associated with the maximum entry of $\bar\bbm_i$ (i.e., $\alpha_i = \|\bar\bbm_i\|_\infty = |\bar{M}_{ik}|$). Then, we have
\begin{equation*}
    \begin{aligned}
        \E{\sqrt{Z_{ip}}} & = \E{\sqrt{\Omega_{kp}\bar{M}_{ik}^2 + \sum_{j\not = i,k}\Omega_{jp}\bar{M}_{ij}^2}}\\
        &\geq \E{\sqrt{\Omega_{kp}\bar{M}_{ik}^2 + \sum_{j\not = i,k}\Omega_{jp}\bar{M}_{ij}^2}\given\Omega_{kp} = 1}\Pr{\Omega_{kp} = 1}\\
         & \stackrel{(a)}{=} \theta \alpha_i \E{\sqrt{1 + \frac{1-\theta}{\theta}t}}\geq \theta \alpha_i \E{\sqrt{e^{t}}}\geq \theta \alpha_i \E{1+\frac{t}{2}}\\
        & = \theta \alpha_i\E{1+\frac{\theta}{2(1-\theta)}  \sum_{j\not = i, k}\frac{\Omega_{jp}\bar{M}_{ij}^2}{\alpha_i^2}}\\
        & \stackrel{(b)}{=}\sqrt{\theta}(1 - \bar\sigma_i)
    \end{aligned}
\end{equation*}
In (a) we let $t = \frac{\theta}{1-\theta}  \sum_{j\not = i, k}\frac{\Omega_{jp}\bar{M}_{ij}^2}{\alpha_i^2} \leq 1$, as $\alpha_i^2\geq\theta$. And in (b), we let $\bar\sigma_i = 1 - \sqrt{\theta}\alpha_i \left[1 + \frac{(1-\alpha_i^2)\theta^2}{2\alpha_i^2(1-\theta)}\right]$.
It can be checked that $f(\alpha_i,\theta) := 1 - \bar\sigma_i = \sqrt{\theta}\alpha_i \left[1 + \frac{(1-\alpha_i^2)\theta^2}{2\alpha_i^2(1-\theta)}\right]\in [\theta+\frac{\theta^2}{2},\sqrt{\theta}]$, for $\alpha_i \geq \sqrt{\theta}$ and $\theta\leq e^{-1}$. As a result, for any $\alpha_i\in[\sqrt{\theta},1]$, we get $\bar\sigma_i = 1-f
(\alpha_i,\theta)\in[1-\sqrt{\theta},1-\theta-\frac{\theta^2}{2}]$. So $\bar\sigma_i\in(0,1)$ is a feasible lower bound for $\alpha_i\geq\sqrt{\theta}$ with $\theta\leq e^{-1}$.   

Note that the feasible range of $\alpha_i$ is $[N^{-1/2},1]$, and we do not know whether $\sqrt{\theta}\leq e^{-1/2}\approx 0.6$ would be within it. So we apply the following strategy to select the lower bound $\sigma_i'$, i.e., 
\begin{equation} \label{eq.lemma3_sigma}
    \sigma'_i = \left\{
    \begin{aligned}
        & \frac{\alpha_i^2}{\theta}, & \alpha_i\in (0,\sqrt\theta]\cap [N^{-1/2},1]\\
        & 1 - f(\alpha_i,\theta), & \alpha_i\in (\sqrt{\theta},1]\cap[N^{-1/2},1].
    \end{aligned}
    \right.
\end{equation}
%

Besides, we have $\E{|\theta \bar{X}^{(\ccalS)}_{ip}|^2} = \E{|\Omega_{ip} \bar{X}_{ip}|^2} = \E{\Omega_{ip}}\E{|\bar{X}_{ip}|^2} = \theta$, so $\E{|\bar{X}^{(\ccalS)}_{ip}|^2} = \frac{1}{\theta}$ as desired. $\hfill\blacksquare$


\noindent\textbf{Proof of Lemma \ref{lemma:2}} Because $\bar{M}_{ii} = 0$, $\bar{X}_{ip} = \sum_{j = 1}^N\bar{M}_{ij}X_{jp} = \sum_{j\not = i}\bar{M}_{ij}\Omega_{jp}\gamma_{jp}\theta^{-1/2}$, for all $(i,p)$. As $\bar\bbm_i$ has unit $\ell_2$ norm, upon conditioning on $Z_{ip}=\sum_{j\not = i}\Omega_{jp}\bar{M}_{ij}^2=z$, $\bar{X}_{ip}$ is distributed as $\textrm{Normal}(0,\frac{z}{\theta})$. Then $\E{| \bar{X}_{ip}|\given Z_{ip}=z} = \beta_0\sqrt{\frac{z}{\theta}}$, and $\E{|\bar{X}_{ip}|}= \theta^{-1/2}\beta_0\E{\sqrt{Z_{ip}}}$. As $Z_{ip}$ here is defined exactly as in the proof of Lemma \ref{lemma:1}, we also have $\sqrt{\theta}(1-\sigma'_i)\leq\E{Z_{ip}}\leq\sqrt{\theta}$, with $\sigma'_i$ defined as in \eqref{eq.lemma3_sigma}.  Notice that we have the same bounds for $\E{|\bar{X}_{ip}|}$ and $\E{|\bar{X}^{(\ccalS)}_{ip}|}$. However, for $\E{|\bar{X}_{ip}|^2}$, the result will be different, as $\E{|\bar{X}_{ip}|^2} = 1$. $\hfill\blacksquare$

From Lemma \ref{lemma:1} and \ref{lemma:2}, we know for an arbitrary hollow-matrix $\bbM$ and Bernoulli-Gaussian distributed $\bbX$ with $\theta\in (0,e^{-1}]$, the product $\bar{\bbX} = \bar\bbM\bbX$ and its masked version $\bar\bbX^{(\ccalS)}$ have bounded entries. Next, Lemma \ref{lemma:exp_tail} establishes their entries have uniform exponential tails.
\begin{mylemma} \label{lemma:exp_tail}
    The elements of both $\bar\bbX$ and $\bar\bbX^{(\ccalS)}$ have uniform exponential tails. Specifically, $\Pr{|\bar{X}_{ip}|>t}\leq e^{-ut+O(u^2)},$ $\forall{u}\in[0,\sqrt{2}\theta/\alpha_i]$, $t\geq 0$. Likewise, $\Pr{|\bar{X}^{(\ccalS)}_{ip}|>t}\leq e^{-ut+O(u^2)},$ $\forall{u}\in[0,\sqrt{2\theta^3}/\alpha_i],$ $t\geq 0$.
\end{mylemma}
\noindent\textbf{Proof of Lemma \ref{lemma:exp_tail}}. To show that $\bar{X}^{(\ccalS)}_{ip}$ has a uniform exponential tail for all $(i,p)$, we consider $0\leq u\leq u_i = \sqrt{2\theta^3}/\alpha_i$,
\begin{align} 
        \E{e^{u \bar{X}^{(\ccalS)}_{ip}}} & = \E{e^{\frac{u\Omega_{ip}}{\theta}\sum_{j\not = i}\bar{M}_{ij}X_{jp}}}\nonumber\\
        & = \prod_{j \not = i} \E{e^{\frac{u\Omega_{ip}\Omega_{jp}}{\theta^{3/2}}\bar{M}_{ij}\gamma_{jp}}}\nonumber\\
        & = \prod_{j \not = i}\mathbb{E}_{\gamma_{jp}}\left[\theta[\theta e^{\frac{u}{\theta^{3/2}}\bar{M}_{ij}\gamma_{jp}}+(1-\theta)]+(1-\theta)\right]\nonumber\\
        &\stackrel{(a)}{\leq} \prod_{j \not = i} \left[\theta\left[\theta(1+u^2\bar{M}_{ij}^2/\theta^3)+(1-\theta)\right]+(1-\theta)\right]\nonumber\\
        & = \prod_{j \not = i} (1 + \theta^{-1} u^2\bar{M}_{ij}^2)\leq \prod_{j \not = i}e^{\frac{1}{\theta}u^2\bar{M}_{ij}^2} = e^{\frac{u^2}{\theta}}.\label{eq:proof_lemma_tail_01}
\end{align}
Notice that in (a), we have applied: (i) $\E{e^{uX}} = e^{\frac{u^2}{2}}, \forall u\geq 0$ when $X\sim \textrm{Normal}
(0,1)$, so $\E{e^{u\theta^{-3/2}\bar{M}_{ij}\gamma_{jp}}} = e^{u\theta^{-3/2}\bar{M}_{ij}^2/2} = e^{\frac{u^2}{2\theta^3}\bar{M}_{ij}^2}$ because $\gamma_{jp}\sim \textrm{Normal}
(0,1)$; and then (ii) $e^\kappa \leq 1+ 2\kappa,\forall{\kappa\in[0,1]}$ as $\kappa = \frac{u^2}{2\theta^3}\bar{M}_{ij}^2\leq \frac{u^2}{2\theta^3}\alpha_i^2 \leq 1$. All in all, $\forall{u\in(0,u_i]}$ and $\forall{t \geq 0}$ we have
\begin{equation*}
     \Pr{\bar{X}^{(\ccalS)}_{ip}\geq t} = \Pr{e^{u\bar{X}^{(\ccalS)}_{ip}}\geq e^{ut}}\leq e^{-ut}\E{e^{u\bar{X}^{(\ccalS)}_{ip}}}\leq e^{-ut+\frac{u^2}{\theta}}.
\end{equation*}
By symmetry, it follows $\Pr{\bar{X}^{(\ccalS)}_{ip}\leq - t}\leq e^{-ut+\frac{u^2}{\theta}}$. So $\bar{X}^{(\ccalS)}_{ip},\forall{(i,p)}$ has a uniform exponential tail, i.e., $\Pr{|\bar{X}^{(\ccalS)}_{ip}|>t}\leq e^{-ut+O(u^2)},$ $\forall{u}\in[0,\sqrt{2\theta^3}/\alpha_i],$ and all $t\geq 0$. 

To show that $\bar{X}_{ip}$ also has a uniform exponential tail for all $(i,p)$, we follow a similar approach and consider $0\leq u\leq u_i = \sqrt{2}\theta/\alpha_i$ to obtain
\begin{align*}
        \E{e^{u\bar{X}_{ip}}} & = \E{e^{u\sum_{j\not = i}\bar{M}_{ij}X_{jp}}}\\
        & = \prod_{j \not = i} \E{e^{\frac{u\Omega_{jp}}{\theta^{1/2}}\bar{M}_{ij}\gamma_{jp}}}\\
        & = \prod_{j \not = i}\mathbb{E}_{\gamma_{jp}}\left[\theta e^{\frac{u}{\theta^{1/2}}\bar{M}_{ij}\gamma_{jp}}+(1-\theta)\right]\\
        &\stackrel{(b)}{\leq} \prod_{j \not = i} \left[\theta(1+u^2\bar{M}_{ij}^2/\theta)+(1-\theta)\right]\\
        & = \prod_{j \not = i} (1 +  u^2\bar{M}_{ij}^2)\leq \prod_{j \not = i}e^{u^2\bar{M}_{ij}^2} = e^{u^2}
\end{align*}
Notice that in (b), we applied the same inequalities as (a) in \eqref{eq:proof_lemma_tail_01}. Then $\forall{u\in(0,u_i]}$ and $\forall{t \geq 0}$,
\begin{equation*}
     \Pr{\bar{X}_{ip}\geq t} = \Pr{e^{u\bar{X}_{ip}}\geq e^{ut}}\leq e^{-ut}\E{e^{u\bar{X}_{ip}}}\leq e^{-ut+u^2}
\end{equation*}
By symmetry, $\Pr{\bar{X}_{ip}\leq - t}\leq e^{-ut+u^2}$ as well. So $\bar{X}_{ip},\forall{(i,p)}$ has a uniform exponential tail, i.e., $\Pr{|\bar{X}_{ip}|>t}\leq e^{-ut+O(u^2)},$  $\forall{u}\in[0,\sqrt{2}\theta/\alpha_i],$ and all $t\geq 0$. $\hfill\blacksquare$

The final ingredient needed to prove Proposition \ref{proposition:2} is Lemma \ref{lemma:gaussian_tail}, which shows that the $\ell_1$ norms of the $i$-th rows of $\bar\bbX$ and $\bar\bbX^{(\ccalS)}$ concentrate to $\beta_i:=\E{|\bar{X}_{ip}|}=\E{|\bar{X}^{(\ccalS)}_{ip}|}$. 
\begin{mylemma} \label{lemma:gaussian_tail}
    Let $\check{x}^{(\ccalS)}_i := \frac{1}{\sqrt{P}}\left(\sum_{p=1}^P |\bar{X}^{(\ccalS)}_{ip}| - P\beta_i\right)$ and $\check{x}_i := \frac{1}{\sqrt{P}}\left(\sum_{p=1}^P |\bar{X}_{ip}| - P\beta_i\right)$. Then both $\check{x}^{(\ccalS)}_i$ and $\check{x}_i$ have sub-Gaussian tails (as defined in \cite[Definition 2.1]{matouvsek2008variants}), up to $\frac{\theta^{3/2}}{\sqrt{2}\alpha_i}\sqrt{P}$ and $\frac{\theta }{\sqrt{2}\alpha_i}\sqrt{P}$, respectively.\\
\end{mylemma}
\noindent\textbf{Proof of Lemma \ref{lemma:gaussian_tail}}. Firstly, note that $\E{\bar{X}_{ip}}=0$ and from Lemma \ref{lemma:2} it follows that $\var{\bar{X}_{ip}}=1,$ $\E{|\bar{X}_{ip}|} = \beta_i$. {Moreover, according to Lemma \ref{lemma:exp_tail}}, $\bar{X}_{ip}$ has an exponential tail, i.e., $\Pr{|\bar{X}_{ip}|>t}\leq e^{-ut+O(u^2)}, \forall{u}\in[0,\sqrt{2}\theta/\alpha_i],$ and all $t\geq 0$. From \cite[Proposition 5.2, Lemma 2.3]{matouvsek2008variants}, it can be shown that $\forall i,$ $\check{x}_i$ has a sub-Gaussian tail up to $\frac{\theta }{\sqrt{2}\alpha_i}\sqrt{P}$, i.e., $\E{e^{u\check{x}_i}}<e^{O(u^2)}$ and $ \E{e^{-u\check{x}_i}}<e^{O(u^2)},\forall u \in\left(0,\frac{\theta }{\sqrt{2}\alpha_i}\sqrt{P}\right)$. 

For $\check{x}^{(\ccalS)}_i$, we also have  $\E{\bar{X}^{(\ccalS)}_{ip}}=0$ and $\E{|\bar{X}^{(\ccalS)}_{ip}|} = \beta_i$. From Lemma \ref{lemma:1} we know the only difference with $\bar{X}_{ip}$ is $\var{\bar{X}^{(\ccalS)}_{ip}}=\frac{1}{\theta}$. So we can still apply \cite[Proposition 5.2, Lemma 2.3]{matouvsek2008variants} and show $\forall i,$ $\check{x}^{(\ccalS)}_i$ has a sub-Gaussian tail up to $\frac{\theta }{\sqrt{2}\alpha_i^{3/2}}\sqrt{P}$, i.e., $\E{e^{u\check{x}^{(\ccalS)}_i}}<e^{O(u^2)}$ and $ \E{e^{-u\check{x}^{(\ccalS)}_i}}<e^{O(u^2)},\forall u \in\left(0,\frac{\theta^{3/2} }{\sqrt{2}\alpha_i}\sqrt{P}\right)$. $\hfill\blacksquare$\\

\noindent\textbf{Proof of Proposition \ref{proposition:2}}. To show (b), from Lemma \ref{lemma:gaussian_tail} we know $\check{x}^{(\ccalS)}_i = \frac{1}{\sqrt{P}} \left(\sum_{p=1}^P|\bar{X}^{(\ccalS)}_{ip} | - \beta_i P\right) =  \frac{1}{\sqrt{P}}\left(\frac{1}{\theta}\|[\bar\bbm_i^\top \bbX]_{\ccalS}\|_1 - P\beta_i\right)$ has a sub-Gaussian tail up to $\frac{\theta^{3/2}}{\sqrt{2}\alpha_i}\sqrt{P}$. As a result, for some $\sigma_2\in(0,1)$, we have 
\begin{align*}
\Pr{\frac{1}{\theta \beta_i P}\|[\bar\bbm_i^\top \bbX]_\ccalS\|_1\geq 1 + \sigma_2} 
= {}& \Pr{\check{x}^{(\ccalS)}_i \geq \beta_i\sqrt{P}\sigma_2}\\
\leq {}& e^{-C'\beta_i^2P\sigma_2^2},
\end{align*}
for some constant $C'>0$. We require $\beta_i\sqrt{P}\sigma_2 \leq\frac{\theta^{3/2}}{\sqrt{2}\alpha_i}\sqrt{P} \Rightarrow \sigma_2\leq \min\left\{ \frac{\theta^{3/2}}{\sqrt{2}\alpha_i\beta_i},1\right\}$. Because $\theta\in(0,e^{-1}]$ and $\frac{1}{\sqrt{2}\alpha_i\beta_i}\geq \frac{1}{\sqrt{2}\beta_0} = \frac{\sqrt{\pi}}{2}$, we have $\frac{\sqrt{\pi}\theta^{3/2}}{2}\leq \min\left\{ \frac{\theta^{3/2}}{\sqrt{2}\alpha_i\beta_i},1\right\}$. So it suffices to select $\sigma_2\in\left(0, \frac{\sqrt{\pi}\theta^{3/2}}{2}\right]$. Now for $e^{-C'\beta_i^2P\sigma_2^2}\leq \frac{\delta}{2}$, we have $P\geq \frac{C}{\sigma_2^2}\log{\frac{2}{\delta}}$, where $C\geq \frac{1}{C' \beta_i^2}$ is some constant. Putting all pieces together we have $\Pr{\frac{1}{\theta \beta_i P}\|[\bar\bbm_i^\top \bbX]_\ccalS\|_1\leq 1 - \sigma_2}\leq \frac{\delta}{2}$ and Proposition \ref{proposition:2}(b) follows.

To show (a), notice that from Lemma \ref{lemma:gaussian_tail} we know $\check{x}_i   =  \frac{1}{\sqrt{P}} (\sum_{p=1}^P|\bar{X}_{ip}| - \beta_i P) = \frac{1}{\sqrt{P}}(\|\bar\bbm_i^\top \bbX\|_1 - P\beta_i)$ has a sub-Gaussian tail up to $\frac{\theta }{\sqrt{2}\alpha_i}\sqrt{P}$. Following similar steps as above for (b), one can readily arrive at Proposition \ref{proposition:2}(a).$\hfill\blacksquare$ 

Having established Proposition \ref{proposition:2}, we have almost all ingredients needed to prove Proposition \ref{proposition:1}. Before doing so,
we introduce a final lemma to show that Bernoulli-Gaussian random variables have bounded-energy. 
\begin{mylemma}[Bounded energy]
\normalfont
\label{lemma:5}
Let $\{X_{ip}\}$ be i.i.d. random variables drawn from the Bernoulli-Gaussian model in in Definition \ref{mydef_BGmodel}, with $\theta\in \left(0,e^{-1}\right]$. For any vector $\bba\in\reals^{N}$ and some $\sigma_3>0$, if $P\geq\frac{\breve{C}}{ \beta_0^2\sigma_3^2}\log{\frac{2}{\delta}}$ where $\breve{C}$ is some constant, we have
\begin{equation} \label{eq:lemma_5}
    \Big|\sum_{i,p} a_i|X_{ip}|\Big| \leq (1+\sigma_3) \left|\E{\sum_{i,p} a_i|X_{ip}|}\right|
\end{equation}
holds with probability at least $1-\delta$.
\end{mylemma}
\noindent\textbf{Proof of Lemma \ref{lemma:5}}. Note that for $ \theta\in\left(0,e^{-1}\right]$, we have $\E{|X_{ip}|}=\sqrt{\theta}\beta_0,$ $\E{|X_{ip}|^2} = 1,$ and $\var{|X_{ip}|} = 1 - \theta\beta_0^2 >0$, for $i\in\{1,\ldots,N\},p\in\{1,\ldots,P\}$. Besides, $|X_{ip}|$ has an upper sub-Gaussian tail since it has the folded Normal distribution and, 
\begin{equation*}
        \Pr{|X_{ip}|\geq t}\leq e^{-\frac{1}{2}t^2},
\end{equation*}
Similar to \cite[Lemma 2.4]{matouvsek2008variants}, one can show $\E{e^{u|X_{ip}|}}\leq e^{\sqrt\theta \beta_0+Cu^2}$, for some constant $C$ and all $u>0$. Now, let $\breve{X}_i:=\frac{1}{\sqrt{P}}\sum_{p=1}^P(|X_{ip}| -\sqrt{\theta}\beta_0)$, so $\E{\breve{X}_i}=0$, $\var{\breve{X}_i} = \E{\breve{X}_i^2} =1$. Following similar ideas as in the proof of Lemma \ref{lemma:gaussian_tail}, one has
\begin{equation*}
        \E{e^{u\breve{X}_i}}  = \prod_{p=1}^P \E{e^{u/\sqrt{P}(|X_{ip}|-\sqrt{\theta}\beta_0)}} \leq (e^{C\frac{u^2}{P}})^P = e^{O(u^2)}.
\end{equation*}
An analogous calculation also yields $\E{e^{-u\breve{X}_i}}\leq e^{O(u^2)}$. 

Let $\bar{a}_i=a_i/\|\bba\|^2_2,$ for $i\in\{1,\ldots,N\}$. From \cite[Lemma 2.2]{matouvsek2008variants}, both $\breve{Y}_1: = \sum_i \bar{a}_i \breve{X}_i$ and $\breve{Y}_2: = -\sum_i \bar{a}_i \breve{X}_i$ have sub-Gaussian tails, i.e., $\Pr{\breve{Y}_1 \geq t}\leq e^{-\breve{C}_1t^2}$ and $\Pr{\breve{Y}_2 \geq t}\leq e^{-\breve{C}_2t^2}$, for some constants $\breve{C}_1,\:\breve{C}_2$. Hence w.p. at least $1 - e^{-\breve{C}_1t^2}$,  we have
\begin{align*}
        \sum_{i=1}^N\frac{\bar{a}_i}{\sqrt{P}}\sum_{p=1}^P(|X_{ip}| -\sqrt{\theta}\beta_0)& < t\\
        \Rightarrow  \sum_{i,p} a_i(|X_{ip}| - \sqrt{\theta}\beta_0)& < \|\bba\|_2^2 \sqrt{P}t\\
        \Rightarrow \sum_{i,p} \left(a_i|X_{ip}| - \E{a_i|X_{ip}|}\right) & < \sigma_3 \E{\sum_{i,p} a_i|X_{ip}|},
\end{align*}
where $\sigma_3: = \frac{\sqrt{P}t}{\E{\sum_{i,p} \bar{a}_i|X_{ip}|}} = \frac{t}{\sqrt{\theta P}\beta_0 }$. If we let $\Pr{\breve{Y}_1 \geq t}\leq e^{-\breve{C}_1t^2}\leq\frac{\delta}{2}$, then we will have $P\geq\frac{\breve{C}_1}{\theta \beta_0^2\sigma_3^2}\log{\frac{2}{\delta}}$. One can also derive the other bound from $\Pr{\breve{Y}_2 \leq -t}\leq e^{-\breve{C}_2t^2}$, i.e., $-\sum_{i,p} a_i|X_{ip}| < - (1 - \sigma_3) \E{\sum_{i,p} a_i|X_{ip}|}$ w.p. at least $1 - e^{-\breve{C}_2t^2}$. To complete the proof, set $\breve{C} = \max(\breve{C}_1, \breve{C}_2)/\theta$.\hfill$\blacksquare$

\noindent\textbf{Proof of Proposition \ref{proposition:1}}. To prove $\hat{\bbw}=\mathbf{1}_N$ is the unique solution to problem \eqref{pb_2}, it is equivalent to say for all feasible perturbations $\bbdelta\in\reals^N$ such that $\bar{\bbr}^\top\bbdelta = 0$, then $\|\ccalP(\mathbf{1}_N + \bbdelta)\bbX\|_{1,1} \geq \|\bbX\|_{1,1}$ holds for any Bernoulli-Gaussian random $\bbX\in\reals^{N\times P}$ with high probability. Here we define the hollow matrix $\ccalM(\bbdelta)$ as $\ccalP(\bbdelta)$ with 0 diagonal elements (i.e., $\ccalM(\bbdelta)_{ij} = \ccalP(\bbdelta)_{ij}, \forall{i\neq j}$; $\ccalM(\bbdelta)_{ii} = 0$, $\forall{i}$). We can compute the sub-gradient of $\|\ccalP(\mathbf{1}_N + \bbdelta)\bbX\|_{1,1} $ at $\bbdelta = \mathbf{0}_N$ as, 
\begin{equation} \label{eq_thm0_step1}
    \begin{aligned}
        \|\ccalP(\mathbf{1}_N + \bbdelta)\bbX\|_{1,1}  \geq \|\bbX\|_{1,1} & + \|[\ccalM(\bbdelta)\bbX]_{\ccalS^c}\|_{1,1}+ \sum_{p=1}^P\bbepsilon_p^\top \ccalP(\bbdelta)\bbx_p\\
        \geq  \|\bbX\|_{1,1} + \|\ccalM(\bbdelta)\bbX\|_{1,1} & - 2\|[\ccalM(\bbdelta)\bbX]_{\ccalS}\|_{1,1}+ \sum_{i,p} |X_{ip}|\inner{\bbv_i\circ\bbv_i,\bbdelta}. 
\end{aligned}
\end{equation}
Applying the Lemma \ref{lemma:5}, we can bound the above last term in \eqref{eq_thm0_step1} as
\begin{align}
        \sum_{i,p}|X_{ip}|\inner{\bbv_i\circ\bbv_i,\bbdelta}  \geq {} & - (1 + \sigma_3)\Bigl|\E{\sum_{i,p}\frac{\Omega_{ip}|\gamma_{ip}|}{\sqrt{\theta}}\inner{\bbv_i\circ\bbv_i,\bbdelta}}\Bigr|\nonumber\\
         ={} & -(1 + \sigma_3)\sqrt{\theta}P\beta_0|\mathbf{1}_N^\top \bbdelta|, \label{proof_prop1_lowerbound1}
\end{align}
where we used $\beta_0:= \E{|\gamma_{ip}|} = \sqrt{\frac{2}{\pi}}$. 

Next,  we are going to bound $\|\ccalM(\bbdelta)\bbX\|_{1,1}  - 2\|[\ccalM(\bbdelta)\bbX]_{\ccalS}\|_{1,1}$ by applying Proposition \ref{proposition:2}. Let $\bbm_i^\top$ be the $i$-th row of the hollow matrix $\ccalM(\bbdelta)$, i.e., $\ccalM(\bbdelta)= [\bbm_1,\ldots,\bbm_N]^\top\in\reals^{N\times N}$. Likewise, let $\bbomega_i^\top$ be the $i$-th row of $\bbOmega$. Then from (a) and (b) in \eqref{prop:hollow_vector}, we have $\|\bbm^\top_i \bbX\|_1\geq (1-\sigma_1)\beta_i P\|\bbm_i\|_2$ and $-\|\bbomega_i^\top\circ (\bbm_i^\top\bbX)\|_1\geq -(1+\sigma_2)\beta_i \theta P\|\bbm_i\|_2$. So we obtain
\begin{align*}
        \|\bbm^\top_i \bbX\|_1 - & 2\|\bbomega_i^\top\circ (\bbm_i^\top\bbX)\|_1  \geq [(1-\sigma_1)-2\theta(1+\sigma_2)]\beta_i P\|\bbm_i\|_2\nonumber\\
         & \geq [(1-\sigma_1)-2\theta(1+\sigma_2)](1-\sigma_4)\beta_0 P\|\bbm_i\|_2,
\end{align*}
where the last inequality holds because  Lemma \ref{lemma:1} asserts that $\beta_i\in[(1-\sigma'_i)\beta_0,\beta_0]$, for some $\sigma'_i\in(0,1)$. Specifically, given a hollow-vector $\bbm_i$ with $\alpha_i = \|\bbm_i\|_\infty/\|\bbm_i\|_2$, the $\sigma'_i$ can be computed via \eqref{eq.lemma3.1}.  Hence, we can let $\sigma_4 = \max\{\sigma'_i\}_{i=1}^N$ and by vertically stacking the row vectors $\bbm_i^\top,\: i\in\{1,\ldots,N\}$, the  bound 
\begin{align}
    &\|\ccalM(\bbdelta)\bbX\|_{1,1} - 2\|[\ccalM(\bbdelta)\bbX]_{\ccalS}\|_{1,1} \nonumber \\
    & \geq [(1-\sigma_1)-2\theta(1+\sigma_2)](1-\sigma_4)\beta_0 P\|\ccalM(\bbdelta)\|_{2,1} \label{proof_prop1_lowerbound2}
\end{align}
holds with probability at least $1-\delta$. Summing \eqref{proof_prop1_lowerbound1} and \eqref{proof_prop1_lowerbound2}, we have
\begin{align}
        & \|\ccalM(\bbdelta)\bbX\|_{1,1} - 2\|[\ccalM(\bbdelta)\bbX]_{\ccalS}\|_{1,1} + \sum_{ip} |X_{ip}|\inner{\bbv_i\circ\bbv_i,\bbdelta}\label{eq_thm0_step2}\\
        & \geq [(1-\sigma_1) - 2\theta(1+\sigma_2)](1-\sigma_4)\Bigl(\|\ccalM(\bbdelta)\|_{2,1} -C_1|\mathbf{1}_N^\top \bbdelta| \Bigr),       \nonumber 
\end{align}
where we defined $C_1 = \frac{(1+\sigma_3)\sqrt{\theta}}{[(1-\sigma_1) - 2\theta(1+\sigma_2)](1 - \sigma_4)}$. When $\theta\leq
\frac{1-\sigma_1}{2(1+\sigma_2)}$, we have $C_1 >0$ and the lower bound will be non-negative if $\|\ccalM(\bbdelta)\|_{2,1} \geq C_1 |\mathbf{1}_N^\top \bbdelta|$. To show such a $\theta$ is feasible, recall that in Proposition \ref{proposition:2} we require that $\theta\in \left(0,e^{-1}\right]$, $\sigma_1\in \left(0,\frac{\sqrt{\pi}\theta^{3/2}}{2}\right]$, and $\sigma_2 \in \left(0,\frac{\sqrt{\pi}\theta}{2}\right]$. Hence we need,
%
%
%
\begin{equation*} 
        \theta  \leq \frac{1-(\sqrt{\pi}/2)\theta^{3/2}}{2(1+\sqrt{\pi}\theta/2)}\:
        \Rightarrow f(\theta):= \sqrt{\pi}\theta^2 + 2\theta +\frac{\sqrt{\pi}}{2}\theta^{3/2} - 1\leq 0.
\end{equation*}
As $f(\theta)$ is monotonically increasing in $\theta\in\left(0,e^{-1}\right]$ and $f(0) = -1< 0$, $f(e^{-1}) > 0$, there exists only one solution $\theta_m\in(0.324,0.325)$ such that $f(\theta_m) = 0$. So the feasible range is  $\theta\in (0,\theta_m]$, or for convenience we let $\theta\in(0,0.324]$.

Going back to establishing the lower bound in \eqref{eq_thm0_step2} is non-negative, note that $\|\ccalM(\bbdelta)\|_{2,1}^2 \geq \|\ccalM(\bbdelta)\|_F^2$ and so it will be sufficient to show $\|\ccalM(\bbdelta)\|_F^2\geq C_1^2 |\mathbf{1}_N^\top \bbdelta|^2$. For convenience, let's decompose $\bar{\bbr} = \frac{c}{N}\mathbf{1}_N + \bbd$ and $\bbdelta = a\mathbf{1}_N + \bbb$, where $\mathbf{1}_N^\top \bbb = \mathbf{1}_N^\top \bbd = 0$. From the constraint $\bar{\bbr}^\top \bbw =\bar{\bbr}^\top (\mathbf{1}_N + \bbdelta) = c $, we know $\bar{\bbr}^\top\bbdelta = 0$ and $a = -\frac{\bbd^\top\bbb}{c}$. So $\bbdelta = -\frac{\bbd^\top\bbb}{c} \mathbf{1}_N+\bbd$. Then we have
\begin{align} 
        \|\ccalM(\bbdelta)\|_F^2 -  C_1^2 |\mathbf{1}_N^\top \bbdelta|^2  = {}& \|\ccalM(\bbb)\|_F^2 -  C_1^2 \left|\frac{\bbd^\top \bbb}{c}\right|^2\nonumber\\
         = {}& \|\bbb\|^2\left(1 - \left\|\begin{bmatrix}
           \Tilde{\bbU} \nonumber\\
           (C_1/c)\bbd^\top
         \end{bmatrix}\bar{\bbb}\right\|_2^2\right)\\
         (\star)\:\:\:  \geq {}& \|\bbb\|^2\left[1 - \left(\sigma_{\text{max}}^2(\Tilde{\bbU}) + \frac{C_1^2}{c^2}\|\bbd\|_2^2\right)\right],\label{eq_thm0_step3}
\end{align}
where we let 
$\tbU := (\bbV \circ \bbV)\bbP_1^\perp\in\reals^{N\times N}$ and
$\bar{\bbb}= \bbb/\|\bbb\|_2$. The inequality in $(\star)$ considers the `worst-case' $\bbd$ should be colinear to the dominant right singular vector of $\Tilde{\bbU}$, i.e., given the SVD $\Tilde{\bbU} = {\hbU}\bbSigma{\hbV}^\top$, then $\bbd/\|\bbd\| = \hbv_1$ is the first column of $\hbV_1$.

To ensure $1 - \left(\sigma_{\text{max}}^2(\Tilde{\bbU}) + \frac{C_1^2}{c^2}\|\bbd\|^2\right) \geq 0$, we bound $\|\bbd\|_2$ as
\begin{equation} \label{eq_thm0_condition}
    \begin{aligned}
        \|\bbd\|_2^2 & \leq \frac{c^2(1 - \sigma_{\text{max}}(\Tilde{\bbU})^2)}{C_1^2}\\
        & =\frac{c^2(1 - \sigma_{\text{max}}(\Tilde{\bbU})^2)[(1-\sigma_1)-2\theta(1+\sigma_2)]^2(1-\sigma_4)^2}{(1+\sigma_3)^2\theta},
    \end{aligned}
\end{equation} 
which is \eqref{theorem0_1}, completing the proof.
$\hfill\blacksquare$


\section{Proof of Theorem \ref{theorem_2}} \label{appendix_theo_2}


When the observations are corrupted by noise $\bbN\in\reals^{N\times P}$, i.e., $\bbY = \ccalP(\tbh_0)\bbX+\bbN\in\reals^{N\times P}$, the optimization problem \eqref{pb_1} can be rewritten as the following equivalent problem 
\begin{equation} \label{pb_noise2}
\hat{\bbw} = \text{arg}\min\limits_{\bbw} \| \mathcal{P}(\bbw)\bbX + \ccalP(\tbg_0 \circ \bbw)\bbN \|_{1,1},\quad \text{s. to }\: \bar{\bbr}^\top \bbw  = c
\end{equation}
with the change of variable $\bbw = \tbg\circ\tbh_0$ we used in \ref{appendix_theo_1}.
Again, $\bar{\bbr}^\top = \bbr^\top\text{diag}(\tbg_0)$ is consistent with \eqref{pb_2}. In this case, the solution to \eqref{pb_noise2} $\hat\bbw$ is not expected to be $\mathbf{1}_N$, so we let $\hat{\bbw} = \mathbf{1}_N+\hat\bbdelta = \mathbf{1}_N + \left(-\frac{\bbd^\top \hat{\bbb}}{c}\mathbf{1}_N + \hat{\bbb}\right) = \left(1 -\frac{\bbd^\top \hat{\bbb}}{c}\right)\mathbf{1}_N + \hat\bbb$ as before, where $\bbd = \bar{\bbr} - \frac{c}{N}\mathbf{1}_N$. Then our goal will be to bound the error of problem \eqref{pb_noise2}, i.e., $\hat\bbw-\mathbf{1}_N$, which is expected to vanish if the observations were not corrupted by noise. Once this goal is achieved, the recovery error of \eqref{pb_1}, i.e., $\hat\tbg -\tbg_0$ can also be bounded since $\hat\tbg -\tbg_0=\tbg_0\circ(\hat\bbw-\mathbf{1}_N )$. 

By jointly considering \eqref{eq_thm0_step1}\eqref{eq_thm0_step2}\eqref{eq_thm0_step3}, and from \eqref{eq_thm0_step3} we can conclude that, if \eqref{eq_thm0_condition} is satisfied, e.g. $\|\bbd\|\leq \frac{c\sqrt{(1 - \sigma_{\max}^2(\Tilde{\bbU}))}}{C_1}= c d_0$, the following holds
\begin{equation}
    \begin{aligned}
        & \|\ccalM(\bbdelta)\|_{2,1} -  C_1 |\mathbf{1}_N^\top \bbdelta|  \geq \sqrt{\|\ccalM(\bbdelta)\|_F^2} -  C_1 |\mathbf{1}_N^\top \bbdelta|\\
        & \geq \sqrt{C_1^2 |\mathbf{1}_N^\top \bbdelta|^2 + \|\bbb\|_2^2\left(1 - \sigma_{\text{max}}^2(\Tilde{\bbU}) - \frac{C_1^2}{c^2}\|\bbd\|_2^2\right)} -  C_1 |\mathbf{1}_N^\top \bbdelta|\\
        & = \|\bbb\|\left[\sqrt{1 - \sigma_{\text{max}}^2(\Tilde{\bbU}) - \frac{C_1^2}{c^2} \|\bbd\|_2^2(1-\sigma_5^2)} - \frac{C_1}{c}\sigma_5\|\bbd\|_2\right],
    \end{aligned}
\end{equation}
%
where $\sigma_5:=\frac{|\bbd^\top\bbb|}{\|\bbb\|\|\bbd\|}\in[0,1]$. 

Let $Q:=\frac{(1+\sigma_3)\sqrt{\theta}}{c}\left[\sqrt{c^2 d_0^2-(1-\sigma_5)^2\|\bbd\|_2^2}-\sigma_5 \|\bbd\|_2\right]$. Note that $Q \geq 0$ and then we have
\begin{equation} \label{eq_noise_lower1}
    \begin{aligned}
        \| \ccalP(\hat\bbw)\bbX &+ \ccalP(\tbg_0 \circ \hat\bbw)\bbN \|_{1,1}  = \| \ccalP(\hat\bbw)[\bbX + \ccalP(\tbg_0)\bbN] \|_{1,1} \\
        & \geq \|\ccalP(\hat\bbw)(\bbX+\bbN^{(S)})\|_{1,1} - \|\ccalP(\hat\bbw)\bbN^{(C)}\|_{1,1}\\
        & \geq \|\bbX+\bbN^{(S)}\|_{1,1} + \beta_0PQ \|\hat\bbb\|_2 - \|\ccalP(\hat\bbw)\bbN^{(C)}\|_{1,1}
    \end{aligned}
\end{equation}
where $\bbN^{(S)}:= [\ccalP(\tbg_0)\bbN]_{\ccalS}$ denotes the sub-matrix of $\ccalP(\tbg_0)\bbN$ that has the same support $\ccalS$ as $\bbX$, and its complement as $\bbN^{(C)}:= [\ccalP(\tbg_0)\bbN]_{\ccalS^C} = \ccalP(\tbg_0)\bbN - [\ccalP(\tbg_0)\bbN]_{\ccalS}$. The last inequality holds as $\bbX$ and $\bbX+\bbN^{(S)}$ have the same support. Next, let's find the upper bound for $\|\ccalP(\hat\bbw)\bbN^{(C)}\|_{1,1}$
\begin{equation} \label{eq_noise_lower2}
    \begin{aligned}
   & \|\ccalP(\hat\bbw)\bbN^{(C)}\|_{1,1}  = \left\|\left(1 -\frac{\bbd^\top \hat{\bbb}}{c}\right)\bbN^{(C)} + \bbV\diag(\hat\bbb)\bbV^\top\bbN^{(C)}\right\|_{1,1}\\
        & \leq \|\bbN^{(C)}\|_{1,1} + (d_0  \|\bbN^{(C)}\|_{1,1}  + \|[\bbN^{(C)}]^\top\bbV\odot\bbV\|_{1\rightarrow 2}) \|\hat\bbb\|_2.
    \end{aligned}
\end{equation}
Note that we expect $\bbw_0 = \mathbf{1}_N$ as the `ideal' noise-free solution of \eqref{pb_noise2}, so $c = \bar{\bbr}^\top\mathbf{1}_N$. For optimality, we should have 
\begin{equation} \label{eq_noise_upper1}
    \begin{aligned}
        \| \ccalP(\hat\bbw)\bbX + \ccalP(\tbg_0 \circ \hat\bbw)\bbN \|_{1,1} & \leq \| \ccalP(\bbw_0)\bbX + \ccalP(\tbg_0 \circ \bbw_0)\bbN \|_{1,1} \\
        & = \|\bbX+\bbN^{(S)}\|_1 + \|\bbN^{(C)}\|_{1,1}.
    \end{aligned}
\end{equation}
%
From \eqref{eq_noise_lower1},  \eqref{eq_noise_lower2} and \eqref{eq_noise_upper1} we find
\begin{equation*} 
    \begin{aligned}
        \|\bbX+\bbN^{(S)}\|_{1,1} & + \|\bbN^{(C)}\|_{1,1}  \geq  \|\bbX+\bbN^{(S)}\|_{1,1} - \|\bbN^{(C)}\|_{1,1} \\
         & + ( \beta_0PQ - d_0  \|\bbN^{(C)}\|_{1,1}  - \|[\bbN^{(C)}]^\top\bbV\odot\bbV\|_{1\rightarrow 2}) \|\hat\bbb\|_2\\
    \end{aligned}
\end{equation*}
and while $\beta_0PQ - d_0  \|\bbN^{(C)}\|_{1,1}  - \|[\bbN^{(C)}]^\top\bbV\odot\bbV\|_{1\rightarrow 2}> 0$, we have 
\begin{equation*}
    \|\hat\bbb\|_2  \leq \frac{2 \|\bbN^{(C)}\|_{1,1}}{\beta_0PQ - d_0  \|\bbN^{(C)}\|_{1,1}  - \|[\bbN^{(C)}]^\top\bbV\odot\bbV\|_{1\rightarrow 2}}.
\end{equation*}
Because $\hat\tbg = \tbg_0 \circ \hat\bbw$, the difference vector between $\hat\tbg$ and the `ideal' ground-truth $\tbg_0$ is $\bbd_g = \hat\tbg - \tbg_0 = \tbg_0\circ\bbw - \tbg_0 = \tbg_0\circ(\hat\bbw-\mathbf{1}_N) = \tbg_0\circ\hat\bbdelta$. Recalling that $\hat\bbdelta = - \frac{\bar{\bbr}^\top\hat\bbb}{c} + \hat\bbb = \left(\bbI_N - \mathbf{1}_N\frac{\bar{\bbr}^\top}{c}\right)\hat\bbb$ and $\beta_0=\sqrt{\frac{2}{\pi}}$, we can bound the $\ell_1$ or $\ell_2$ norms of $\bbd_g$ as
\begin{align*}
        \|\bbd_g\|_l & = \|\tbg_0\circ\hat\bbdelta\|_l\\
        & = \left\|\diag(\tbg_0) \left(\bbI_N - \mathbf{1}_N\frac{\bar{\bbr}^\top}{c}\right)\hat\bbb\right\|_l\\
        &\leq \left\|\diag(\tbg_0) \left(\bbI_N - \mathbf{1}_N\frac{\bar{\bbr}^\top}{c}\right)\right\|_{l\rightarrow 2}\|\hat\bbb\|_2\\
        & \leq \frac{2\left\|\diag(\tbg_0) \left(\bbI_N - \mathbf{1}_N\frac{(\bbr\circ\tbg_0)^\top}{c}\right)\right\|_{l\rightarrow 2}\|\bbN^{(C)}\|_{1,1}}{\sqrt{\frac{2}{\pi}}P Q - d_0  \|\bbN^{(C)}\|_{1,1}  - \|[\bbN^{(C)}]^\top\bbV\odot\bbV\|_{1\rightarrow 2}},
\end{align*}
where $\|\cdot\|_l$ stands for the $\ell_1$ and $\ell_2$ norms when $l=1,2$, respectively.
$\hfill\blacksquare$


\section{Epinions Data Sampling and Pre-Processing}\label{appendix_data}


\subsection{Sampling}

Here describe the implemented sampling design applied to $\bbY_{\textrm{obs}}$. The goals of data sampling are: (i) the resulting observation density should be as high as possible; and (ii) all of the users in the sampled dataset should be connected. To achieve those goals, for $k=1,2,\ldots$ we repeat the following three steps. \emph{Step 1:} From the previous item set $\ccalI_k$ we pick all items that have been rated by at least $N_{\min}=150$ users from $\ccalU_k$, and let the new item set be $\ccalI_{k+1}$. \emph{Step 2:} From the previous user set $\ccalU_k$ we pick all users who have rated at least $N_{\min}$ items in $\ccalI_{k+1}$ and let the new user set be the user set candidate $\ccalU'_{k+1}$.  \emph{Step 3.} To maintain the connectivity of the user network, we randomly select a user from $\ccalU'_{k+1}$ and collect all of the users (within $\ccalU'_{k+1}$) that are accessible from this user (including the selected user) to generate the new user set $\ccalU_{k+1}$. The above three steps are repeated until there is no feasible update for $\{\ccalI_k,\ccalU_k\}$.

\subsection{Pre-processing}

As hinted by the results in Theorem \ref{theorem_1}, the unbiasedness of the graph signals is crucial to satisfactory performance of the proposed approach. Hence we centered the rating matrix by adjusting the range from $[1,5]$ to $[-2,2]$.  

%
\small{
\bibliographystyle{IEEEtran}
\bibliography{citations.bib}
}
\end{document}